\documentclass[submission, Phys]{SciPost}

\usepackage[american]{babel}
\usepackage{bm}
\usepackage{graphicx}
\usepackage{amsmath}
\usepackage{amsfonts}
\usepackage{amssymb}
\usepackage{amsthm}
\usepackage{amstext}
\usepackage{amsbsy}
\usepackage{amsopn}
\usepackage{amscd}
\usepackage{amsxtra}
\usepackage{sidecap}

\newtheorem{theorem}{Theorem}
\newtheorem{lemma}{Lemma}
\newtheorem{proposition}{Proposition}
\newtheorem{definition}{Definition}

\newcommand{\n}{\mathfrak{n}}
\newcommand{\Mc}[1]{\mathcal{#1}}

\newcommand{\setN}{\mathbb{N}}
\newcommand{\setR}{\mathbb{R}}

\newcommand{\setC}{\mathbb{C}}
\newcommand{\Id}{\mathbb{I}}

\newcommand{\braket}[2]{\langle #1| #2 \rangle }
\newcommand{\ii}{\textsl{i}}

\DeclareMathOperator*{\sign}{sign}

\begin{document}

\begin{center}{\Large \textbf{
Calculus of Variation and Path-Integrals with Non-Linear Generalized Functions
}}\end{center}

\begin{center}
Quentin Ansel\textsuperscript{1*},
\end{center}

\begin{center}
Laboratoire Interdisciplinaire Carnot de Bourgogne, CNRS UMR 6303, Universit\'{e} de Bourgogne, BP 47870, F-21078 Dijon, France
\\
* quentin.ansel@u-bourgogne.fr
\end{center}


\begin{center}
\today
\end{center}


\section*{Abstract}
{\bf
The calculus of variation and the construction of path integrals is revisited within the framework of non-linear generalized functions. This allows us to make a rigorous analysis of the variation of an action that takes into account the boundary effects, even when the approach with distributions has pathological defects. A specific analysis is provided for optimal control actions, and we show how such kinds of actions can be used to model physical systems. Several examples are studied: the harmonic oscillator, the scalar field, and the gravitational field. For the first two cases, we demonstrate how the boundary cost function can be used to assimilate the optimal control adjoint state to the state of the system, hence recovering standard actions of the literature. For the gravitational field, we argue that a similar mechanism is not possible. Finally, we construct the path integral for the optimal control action within the framework of generalized functions. The effect of the discretization grid on the continuum limit is also discussed.
}

\vspace{10pt}
\noindent\rule{\textwidth}{1pt}
\tableofcontents\thispagestyle{fancy}
\noindent\rule{\textwidth}{1pt}
\vspace{10pt}

\section{Introduction}
\label{sec:intro}
The theory of distribution has been initiated by Laurent Schwartz~\cite{schwartz_theorie_1966,friedlander1998introduction} to provide a rigorous mathematical framework to Dirac's "$\delta$" and many other singularity problems in the analysis of differential equations~\cite{lutzen2012prehistory}. Since it has been one of the most successful mathematical revolutions in physics with applications in electromagnetism~\cite{idemen2011discontinuities}, general relativity~\cite{parker1979distributional,heinzle2002remarks}, and quantum theory~\cite{schwartz1968application,greiner1996field}. It has been also an unavoidable tool in probability and various mathematical modeling.

However, the theory of distribution has a major drawback: the product of two arbitrary distributions is usually ill-defined~\cite{schwartz_theorie_1966,schwartz1954limpossibilite}. This is a big issue in many fields of physics, such as quantum field theory, where the product of Dirac distributions must be used in some calculations. Various directions have been attempted to remedy the problem, and nowadays, the most popular approach is the theory of non-linear generalized functions~\cite{gsponer_concise_2006,colombeau_elementary_2011,colombeau_topics_2011,colombeau_nonlinear_2014}. These "functions" are, in fact, not functions in the usual sense, they are rather a generalization of distributions where the product of two elements is always well-defined. The root of the theory is similar in spirit to the one of distribution, a generalized function can be seen as the limit of a sequence of smooth functions of compact support (or a smooth function with a sufficiently fast decay at infinity). The product of such functions is always well defined and in the limit, it gives us a generalized function. However, if two generalized functions can be assimilated with two distributions, their product may not be a distribution. The set of generalized functions is therefore bigger than the set of distributions. Besides a few technical details, this simple extension removes all the pathological cases encountered with distributions and it opens the door to the analysis of many mathematical models which have been disregarded up to date~\cite{gsponer_concise_2006,colombeau_heisenberg-pauli_2008,grosser2013geometric}.

Among all the fields of applications, the calculus of variations of a functional~\cite{Courant_Hilber_vol_1,Courant_Hilber_vol_2,gelfand2000calculus}, written in its modern version, is a subject that uses the properties of distributions  extensively~\cite{greiner1996field}. However, as useful as it can be, distributions cannot take into account the variations of a functional on the boundary of its integration domain, since it requires computing the product of distributions. We are then forced to impose by hand that the boundary has no effect, or we have to use old approaches, which are not completely satisfactory, because of some divergent quantities. At least in the physics community, these subtleties are totally disregarded and never discussed. However, they play an important role, specifically in quantum field theory and in general relativity. Generalized functions can be used to remedy the problems, and they allow us to rewrite the entire theory consistently. All the pathological details are removed, the product of distributions is not a problem anymore, and the divergences can be interpreted easily as generalized functions. Rewriting the calculus of variations with this new language is the first goal of this paper.

The second goal of this paper is to revisit the modeling of quantum systems with an action similar to the one used in optimal control theory~\cite{ito2008lagrange,contreras_dynamic_2017}, a research program initiated in~\cite{ansel2022loop}, which was motivated by quantum gravity issues. In~\cite{ansel2022loop} a path integral was derived from the optimal control action \cite{bryson1975applied,contreras_dynamic_2017,ito2008lagrange} with starting point, and it has been underlined that it can be used for calculating the time evolution of a quantum state in a straightforward way. The underlying idea is that equations of motion are not included in an action with a quadratic Lagrangian, but with an optimal control action, which has a linear Lagrangian~\cite{contreras_dynamic_2017}. The particularity of this Lagrangian is that the equations of motions are included with the use of Lagrange multipliers called Optimal control adjoint states. The advantage of this choice is that when the path integral is computed, we naturally introduce, in the calculations, a series of Dirac distributions that acts as a propagator of the equations of motion. Hence, in the continuum limit, the path integral can be expressed with the flow of the classical system and the calculation is trivial once the flow is known. The quantum theory is introduced with the use of coherent states that can be computed on the boundary of the integration domain of the action. In quantum gravity these properties are welcome because it allows us to analyze some quantum effects without being able to solve the full Weehler-DeWitt equation~\cite{rovelli_quantum_2007}, we only need to know the flow of a set of well-chosen coherent states.

In this paper, we come back to these ideas, and several aspects of the theory are refined. In particular, we show how the terminal cost of an optimal control problem can be used to include both the quantum scalar product and a constraint on the optimal control adjoint state that forces this latter to take the value of the system state on an extremal. Such an analysis is made possible with the use of generalized functions. Also, the construction of the path integral is revisited with the formalism of generalized function. This allows us to handle at the same time the discretization of the system, as required by the path integral, and the fact that the fields are continuous functions. Moreover, generalized functions allow us to analyze the continuum limit precisely. The theory is illustrated with the analysis of several case studies: the harmonic oscillator~\cite{messiah1962quantum}, the scalar field~\cite{itzykson2012quantum}, and the gravitational field~\cite{misner_gravitation_1973}. The analysis of the harmonic oscillator is mostly a matter of consistency checks, and once the main equations are derived, it is clear that all the expected behaviors follow, both from the classical and quantum sides. This is far less trivial for the other systems. In the case of the scalar field, the construction of well-defined quantum states is not a simple issue, specifically in curved space-time~\cite{gerard2018introduction,khavkine2015algebraic}. In the case of gravity, described by Lorentzian manifold, the construction of coherent states is a harder problem~\cite{rovelli_quantum_2007,rovelli_covariant_2014,stottmeister_coherent_I_2016,bahr_gauge-invariant_2009}. To provide a complete and rigorous description of these problems within the formalism described in this paper, each of these difficult problems can deserve their own articles. Hence, the scope of this paper is limited to a few properties of the system. In particular, it is focused on (i) the description of the system in terms of optimal control action and in terms of non-linear generalized functions (ii) the introduction of a boundary cost function that provides a constraint on the optimal control adjoint state. Issues such as the construction of the quantum Hilbert space and its connection with the boundary cost function are not investigated. 

The paper is organized as follows, in Sec.~\ref{sec:intro_generalized_fun}, we introduce state-of-the-art notions of generalized functions, and a specific extension to tensor fields is proposed. In Sec.~\ref{sec:action_variation_with_GF}, the variation of an action with generalized functions is investigated, and the case of the optimal control action is investigated in detail. In Sec.~\ref{sec:anlysis_action_physical_syst}, the theory is applied to the physical systems mentioned in the previous paragraph. In Sec.~\ref{sec:path_integral_oc}, the path integral for an optimal control action is constructed, and the application of the physical systems of Sec.~\ref{sec:anlysis_action_physical_syst} is discussed. In Sec.~\ref{sec:physical interpretation}, a physical interpretation of the optimal control adjoin state is given and discussed. A conclusion and prospective views are given in Sec.~\ref{sec:conclusion}. In the appendices, the physical interpretation of the adjoint state is made, and the path integral with the standard action of the harmonic oscillator is revisited with the use of non-linear generalized functions.

\section*{List of Symbols}

\begin{itemize}
\item $\eta$ is an arbitrary mollifier, $\eta_q$ is a mollifier of order $q$, and $\eta_\epsilon (x) = \tfrac{1}{\epsilon}\eta(x/\epsilon)$ is a re-scalled mollifier.
\item $\Mc A_q$ is the space of  mollifiers of order $q$.
\item $\Mc S$ is the Schwartz space.
\item $\Mc D$ is the space of test function.
\item $\Id_A(x)$ is the indicator function over $A$, equal to one if $x\in A$, and zero otherwise.
\item $\tilde f$ is a function convoluted by a mollifier.
\item $\Mc E_m$ is the space of moderate function.
\item $\Mc N$ is the space of negligible function, or depending on the context, a normalization constant.
\item $\Mc G$ is the Colombeau algebra.
\item $\mathfrak{n_m}$ is a negligible function of order $m$.
\item $\hat D(x,X)$ is the diffeomorphism operator from a coordinate system $x$ to a coordinate system $X$.
\item $\tfrac{\delta I}{\delta \tilde f(x)}$ is the functional derivative of $I$ with respect to $\tilde f(x)$.
\item $G$ is a Lie group and $\mathfrak{g}$ is its Lie-algebra.
\item $S$ is an action, $S_{quad}$ is the standard quadratic action of a physical system, $S_{OC}$ is the optimal control action of a physical system, and $S_{2nd}$ is an action with second order derivatives.
\item $\hat W$ is the propagator of quantum states, given by the path integral, and $W^\mu$ is the flow associated with the classical equation of motion for a rank-1 tensor.
\end{itemize}

\section{Summary on Non-linear Generalized Functions}
\label{sec:intro_generalized_fun}

This section is dedicated to a pedagogical introduction to generalized functions. Further details can be found in standard textbooks or research articles \cite{gsponer_concise_2006,colombeau2000new,colombeau_elementary_2011,colombeau_topics_2011,colombeau_nonlinear_2014}.

We start our discussion with the generalized function from $\setR$ to $\setR$. Further generalizations are given in the second step.

\subsection{Generalized Function on $\setR$}
The key concept in the definition of generalized functions is the concept of mollifier.
\begin{definition} \textbf{(Mollifier)}
A mollifier $\eta$ of order $q$ is a smooth function $\setR \rightarrow \setR$ of the Schwartz Space $\Mc S$ or the space of test function $\Mc D$ such that:
\[
\int \eta(x) dx =1
\]
\[
\int x^n \eta(x)dx = 0 ~;~ \forall n ~|~1\geq n \geq q.
\]
\label{def:mollifier}
\end{definition}

As explained in detail in the following paragraphs, mollifiers behave like Dirac distributions, and they are the main building blocks for the construction of generalized functions.

\begin{definition} \textbf{(Space of Mollifier)}
The space $\Mc A_q$ of mollifiers of order $q$ is a subset of $\Mc S$ (Schwartz space)or $\Mc D$ (space of test functions) that contains all mollifiers of a given order.
\end{definition}

\begin{proposition}
The set $\Mc A_q$ is non-void for $q=1,2...$
\end{proposition}

The proof is given by the explicit construction of examples for an arbitrary value of $q$. The idea is the following, first, we find by hand a simple example of a mollifier of low order, and then, we construct higher order mollifiers with a recursion formula.

Let us choose a function $\eta_0 \in \Mc D$ such that $\int \eta_0(x)dx =1$. Since there exist functions that satisfy this property, it proves the proposition for $q=0$. Now let us construct a function $\eta_1 = \eta_0 + \alpha_1 \partial_x\eta_0$. The goal is to find $\alpha_1$ so that $\eta_1  \in \Mc A_1$. Since $\eta$ is a test function, we have:
\[
\int x \partial_x \eta(x)dx = -\int  \eta(x)dx = -1 .
\]
Therefore, it suffices that
\[
\alpha_1 = \int x \eta_0(x) dx .
\]
In order to have $\int x \eta_1(x) dx = 0$. Following the same idea, we can construct a mollifier of order 2, $\eta_2$, from the mollifier $\eta_1$, and by recursion, we can construct a mollifier of any arbitrary order. The recursion formula is:
\begin{align}
\label{eq:eta_q_from_eta_q-1}
\eta_q &= \eta_{q-1}+ \alpha_q \partial_x^q \eta_0 \\ 
\alpha_q &= -\frac{1}{n!}\int x^q \eta_{q-1}(x)dx ,
\label{eq:def_alpha_q}
\end{align}
with $\eta_{q-1} \in \Mc A_{q-1}$. 

\begin{definition}\textbf{Family of mollifiers}
A family of mollifiers is the ensemble of mollifiers $\eta_q$ generated from the mother function $\eta_0$ and its derivatives.
\end{definition}

\paragraph{Note:} The underscore index $q$ that specifies the order of the mollifier is usually dropped unless it is specifically necessary.

 As examples, we consider the family of mollifiers generated by the bump function and the Gaussian function.
 
 \paragraph*{Bump mollifier}

The bump function is defined by:
\begin{equation}
\eta_0(x) = A_p e^{\frac{1}{x^p-1}}\Id_{]-1,1[}(x) .
\end{equation}
The parameter $p \in 2 \setN\*$ sets the degree of flatness around zero. In the rest of this paragraph, we use $p=2$. The normalization factor $A_p$ depends on $p$. For $p=2$, it is expressed in terms of a Whittaker's function: $A_2 = \sqrt{\pi/e}W_{-1/2,1/2}(1)=0.44399...$, and $\Id_A$ is the indicator function on $A$. Since this function is positive and symmetric around $0$, it is already a mollifier of order 1. A straightforward application of Eq.~\eqref{eq:eta_q_from_eta_q-1} and Eq.~\eqref{eq:def_alpha_q} gives us $\eta_q = \eta_0(x)\sum_{n=0}^q \alpha_n P_n (x)$. For the first orders, we have
\begin{align*}
\alpha_0 & = 1 & P_0  =& 1 \\
\alpha_2 & \approx  -0.07905681811965157 & P_2  =&  \frac{-2 + 6 x^4}{(-1 + x^2)^4}\\
\alpha_4 & \approx 0.004042404740618485 & P_4  =& \frac{4 (-3 + 6 x^2 + 58 x^4 - 132 x^6 + 45 x^8 + 30 x^{10})}{(-1 + x^2)^8} .
\end{align*}
with $\alpha_{2n+1}=0$. Thus, $\eta_{2n + 1} = \eta_{2n}$.

The functions are plotted in Fig.~\ref{fig:example_mollifier} a). Note that the bump function has compact support, and thus, the repeated derivatives create large picks located near the boundary of the support of definition.

\paragraph*{Gaussian mollifier}
Next, we consider the mollifier
\begin{equation}
\eta_0(x) = \frac{1}{\sigma\sqrt{2\pi}}e^{-x^2/2\sigma^2} .
\end{equation}
Similarly to the bump function, the normalized Gaussian is also a mollifier of order 1, and we also have $\eta_q = \eta_0(x)\sum_{n=0}^q \alpha_n P_n (x)$.
First order coefficients $\alpha_n$ and polynomials $P_n$ are :

\begin{align*}
\alpha_0 & = 1 &  P_0 & = 1\\
\alpha_2 & =  -\sigma^2/2 & P_2 & =  \frac{(x - \sigma) (x + \sigma)}{\sigma^4}\\
\alpha_4 & = \sigma^4/8 & P_4 & = \frac{x^4 - 6 x^2 \sigma^2 + 3 \sigma^4}{\sigma^8} \\
\alpha_6 & = -\sigma^6/48 & P_6 & = -\frac{x^6 - 15 x^4 \sigma^2 + 45 x^2 \sigma^4 - 15 \sigma^6}{\sigma^{12}}  \\
\alpha_8 & =  -\sigma^8/384 & P_8 & =  \frac{x^8 - 28 x^6 \sigma^2 + 210 x^4 \sigma^4 - 420 x^2 \sigma^6 + 
 105 \sigma^8}{\sigma^{16}} ,
\end{align*}
and by the symmetry of the function, any odd order coefficients are 0, thus $\eta_{2n+1} = \eta_{2n}$.

First iterations up to the order 8 are plotted in Fig.~\ref{fig:example_mollifier} b). Interestingly, we recover a shape similar to a sinc function multiplied by an apodization function.

\begin{figure}
{\centering
a) \includegraphics[width =0.45\textwidth]{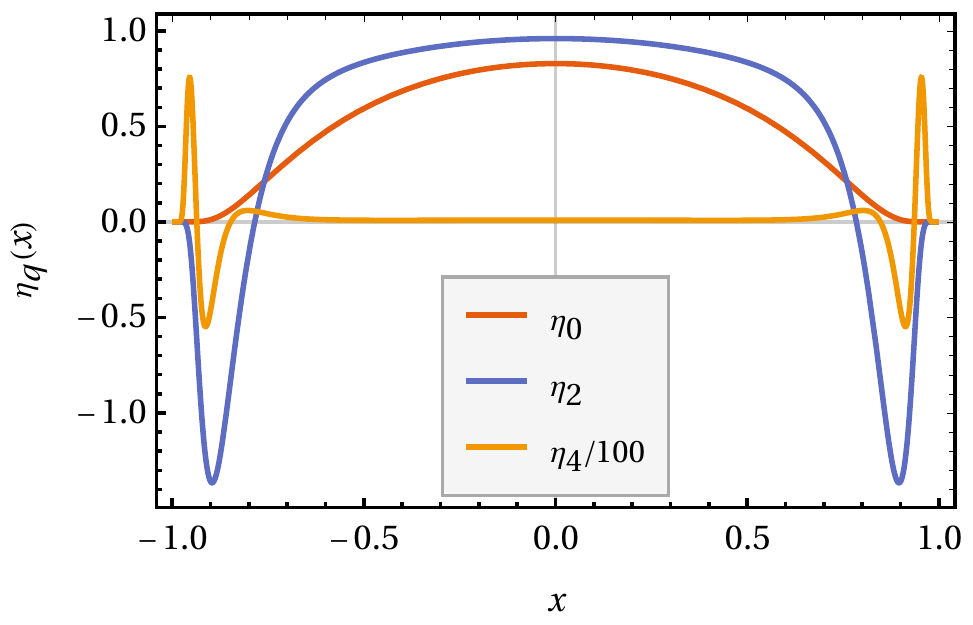}
b) \includegraphics[width =0.45\textwidth]{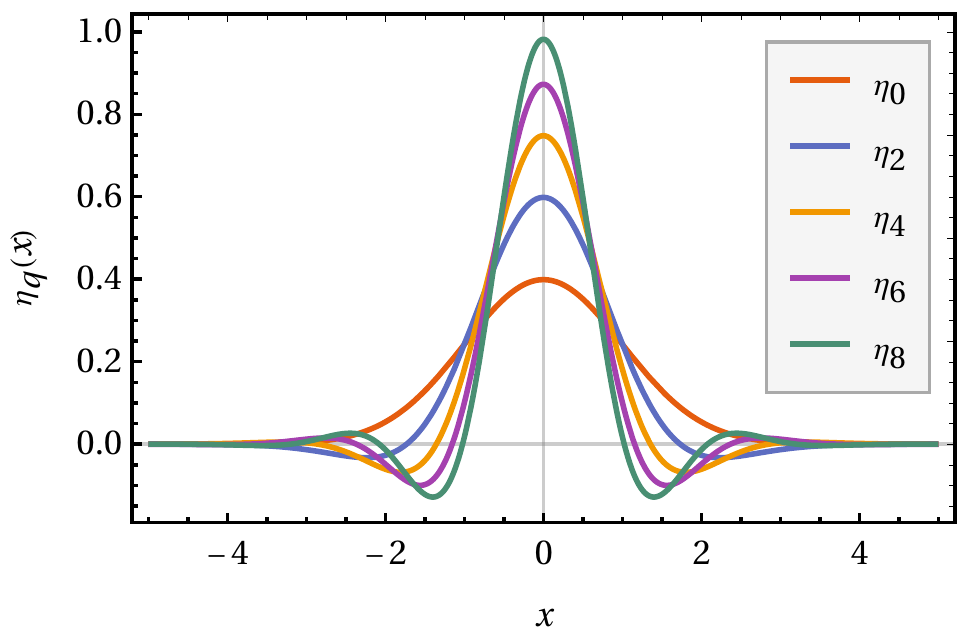}
}
\caption{\textbf{Panel a)} Family of mollifier generated by the bump function $\eta_0(x) = A_p e^{\frac{1}{x^p-1}}\Id_{]-1,1[}(x).$ for $q = 0,2,4$. \textbf{Panel b)} Familly of mollifiers generated by the normalized Gaussian function  $\eta_0(x) = \frac{1}{\sigma\sqrt{2\pi}}e^{-x^2/2\sigma^2} $ for $q=0,2,4,6,8$. In both cases, mollifiers of order $q+1$ are the same as the ones of order $q$.}
\label{fig:example_mollifier}
\end{figure}

Now that the notion of mollifiers has been clarified, we can construct the generalized functions.

\begin{definition} Let $f \in C^{\infty}$, $\eta$ a mollifier of order $q$, and $\epsilon \in ]0,1[$. We define the $\eta$ regularization of $f$ with the convolution product:
\[
\tilde f (x) = \int f(y) \frac{1}{\epsilon} \eta \left( \frac{y - x}{\epsilon}\right) dy
\]
\end{definition}

\paragraph{Note:} In the rest of this paper, the resized mollifier $\frac{1}{\epsilon} \eta \left( \frac{y - x}{\epsilon}\right)$ is often noted $\eta_\epsilon(y-x)$, not to be confused with a $\eta_q$, a mollifier of order $q$.

\begin{proposition}
\[
\tilde f (x) = f(x) + O(\epsilon^{q+1})
\]
\end{proposition}

\begin{proof}
This is a straightforward calculation using the definition \ref{def:mollifier}.
\begin{equation}
\begin{split}
f_\varepsilon(x) &= \int dy \frac{1}{\varepsilon}\eta\left(\frac{y-x}{\varepsilon}\right)f(y) \\
& = \int dz ~\eta\left(z\right) f(x+\varepsilon z)\\
& = f(x) \int dz~ \eta(z) + \sum_{n=1}^\infty \frac{\varepsilon^n \partial_x^{(n)}f(x)}{n!} \int dz~z^n \eta (z) \\
& = f(x) + \sum_{n=q+1}^\infty \frac{\varepsilon^n \partial_x^{(n)}f}{n!} \int dz~z^n \eta (z) \\
& = f(x) + O(\epsilon^{q+1})
\end{split} 
\end{equation}
\end{proof}

With this proposition, we see that mollifiers are similar to a Dirac distribution, but with the difference that the equality holds with a rest of order $\epsilon^{q+1}$. Consequently, there are two possible ways to recover a Dirac distribution in the sense of distribution theory. We can take the limit $\epsilon \rightarrow 0$ or the limit $q \rightarrow \infty$.\footnote{Note that there is a simple way to define a mollifier of order $\infty$. It is sufficient to take a smooth plateau function of compact support, which is equal to 1 on an interval, and to take its Fourier transform. By the derivative properties of the Fourier transform, we obtain that all moments of the function vanish exactly.} In any case, we can consider a class of equivalence of function $\tilde f (x)$ which are equal to $f(x)$ up to a negligible rest of order $o(\epsilon^{q+1})$. The interesting property is that $\tilde f(x) \tilde g (x) = f(x) g(x)$ + a rest of order at least $\epsilon^{q+1}$. Thus, if we consider the ensemble of all classes of equivalence, the product of functions of the ensemble remains in the ensemble. Additionally, the ensemble has all the desired properties to recover Schwartz's distributions.

More precisely, we define the following sets:

\begin{definition}
\textbf{(space of moderate functions) }
\[
\Mc E_m = \left\lbrace \tilde f ~ \vert ~ \forall \alpha \in \setN_0, ~ \exists N \in \setN \text{ with } \sup_{x} |\partial^\alpha \tilde f (x)| = o(\epsilon^{-N}) \text{ as } \epsilon \rightarrow 0 \right\rbrace
\]
Note that for simplicity, we have omitted the support of $\tilde f$, which may induce slight differences if we work in $\Mc D$ or $\Mc S$.
\end{definition}

\begin{definition}
\textbf{(Set of negligible functions) }
\[
\Mc N = \left\lbrace \tilde f ~ \vert ~ \forall \alpha \in \setN_0, ~ \exists m \in \setN \text{ with } \sup_{x} |\partial^\alpha \tilde f (x)| = o(\epsilon^{m}) \text{ as } \epsilon \rightarrow 0 \right\rbrace
\]
In the following, we note a negligible function of order $m$, $\n_m$.
\end{definition}

\begin{definition}
\textbf{(Colombeau algebra) } The Colombeau algebra is the set of all equivalence classes of moderate functions which are equivalent modulo a negligible function, i.e. 
\[
\Mc G = \frac{\Mc E_M}{\Mc N}.
\]
\end{definition}
Elements of $g$ are called non-linear generalized functions.

Let us denote the equivalent class with brackets $[~]$, such that $[g]=\{g + \n_m ~ |~ \n_m \in \Mc N\} \in \Mc G$, with $m$ fixed. We have the properties:
\begin{itemize}
\item $[g]\cdot [h] = [gh]$
\item $D[g] = [Dg]$ with $D$ any partial differential operator.
\end{itemize}

\begin{definition}
\textbf{(Associated distribution)} A generalized function $f$ is said associated with a distribution $T \in \Mc D'$ (or $\Mc S'$) if for any representative function $\tilde f$ we have:
\[
\forall \phi \in \Mc D, ~ \lim_{\epsilon \rightarrow 0}\int \tilde f (x) \phi (x) dx = \braket{T}{\phi} ,
\] 
and we write $f \asymp T$.
\end{definition}
\begin{theorem}
\textbf{(Colombeau local structure theorem)} Any distribution is locally a generalized function~\cite{colombeau2000new}.
\end{theorem}
In order to limit the number of different notations, we omit the square brackets and simply note $g$ instead of $[g]$, when no confusion is possible.

\begin{theorem}
\textbf{(theorem of approximation)}
Let $f:\setR \rightarrow \setR$. We can approximate the function using the following scheme:
\[
f(x) = \lim_{\epsilon \rightarrow 0} \lim_{N \rightarrow \infty}  \sum_{k=1}^N f_k ~\eta_\epsilon (y_k-x).
\]
\label{thm:approx}
\end{theorem}

\begin{proof}
This is a simple decomposition of the convolution product using the Riemann integral:
\begin{equation}
\begin{split}
f(x) &= \lim_{\epsilon \rightarrow 0} \int f(y)~  \eta_\epsilon(y-x) dy \\
& = \lim_{\epsilon \rightarrow 0} \lim_{N \rightarrow \infty}  \sum_{k=1}^N f(y_k) ~ \eta_\epsilon(y_k-x) \Delta y_k  \\
& = \lim_{\epsilon \rightarrow 0} \lim_{N \rightarrow \infty}  \sum_{k=1}^N f_k ~ \eta_\epsilon (y_k-x).
\end{split}
\end{equation}
\end{proof}
Note that imposing $\Delta y_k = \epsilon$ allows us to simplify the expression into : $f(x) \approx \sum_k f(y_k) \eta((y_k-x)/\epsilon)$. 

The quality of the approximation depends on the number of points, the value of $\epsilon$ but also the choice of mollifiers. In practice, to obtain a good approximation of the function, we have to adapt the mollifier and its order with the variations of the function.

\begin{theorem}
\label{thm:scalar_prod_2_Dirac_distrib}
\textbf{(Scalar product of two Dirac distributions)}
Let $\delta_x (z) \equiv \delta_0 (z-x)$ and $\delta_y (z) \equiv \delta_0 (z-y)$ be two Dirac distributions over $\setR$. The scalar product between them is not well defined as a map $\setR \times \setR \mapsto \setR$ but as a map $\setR \times \setR \mapsto \Mc G$, i.e.
\[
\braket{\delta_x}{\delta_y} = \delta_0 (x-y).
\]
\end{theorem}
\begin{proof}
We express the Dirac Distribution with a limit of mollifiers and we compute the "scalar product":
\begin{equation*}
\begin{split}
\lim _{\epsilon,\epsilon'\rightarrow 0}\int dz~ \frac{1}{\epsilon }\eta \left( \frac{z - x}{\epsilon}\right) \frac{1}{\epsilon' }\eta' \left( \frac{z - y}{\epsilon'}\right) & = \lim _{\epsilon,\epsilon'\rightarrow 0} \frac{1}{\epsilon' } \eta' \left( \frac{x - y}{\epsilon'}\right) + O(\epsilon^{q+1}) \\
& = \lim _{\epsilon,\epsilon'\rightarrow 0} \frac{1}{\epsilon }\eta \left( \frac{x - y}{\epsilon}\right) + O(\epsilon '^{q+1})
\end{split}
\end{equation*}
It is then clear that the result is a moderate function which is associated with a Dirac distribution when $\epsilon, \epsilon' \rightarrow 0$. We can also define a generalized function from this result by introducing a class of equivalence with both $\epsilon$ and $\epsilon'$ as parameters defining negligible functions.
\end{proof}

\subsection{Generalized Tensor Fields}

Now that the main concepts of generalized functions are reviewed in the case of function over $\setR$, we can extend their definition to tensor fields. The extension of generalized function from $\setR^n$ to $\setR^m$ is not particularly a problem, however, some subtleties arrive with the notion of diffeomorphism invariance. This is because linear diffeomorphisms of $\setR^n$ are used to define the mapping from $\Mc C^k \mapsto \Mc D$.

The search for a theory of non-linear generalized tensor fields on manifolds has been an active field of research during the last years. Significant progress has been obtained by Vikers, Wilson, Nigsch, and their colleagues~\cite{grosser2013geometric,vickers1998nonlinear,VICKERS2012692,nigsch2020nonlinear}.

The main idea of their theory is to use a smoothing kernel that parallel transport the basis vector of the tangent space at an arbitrary point of the manifold (see \cite{nigsch2020nonlinear} for the most recent version of this latter). Explicitly the smoothing of a vector field $X$ is given by:
\begin{equation}
\tilde X_a (x) = \int_{y \in \Mc M} X_b(y) \Mc T^b_a(x,y) \omega_x (y)
\end{equation}
with $ \Mc T(x,y)$ a two-point tensor field that parallel transport the tangent spaces between $x$ and $y$ using geodesics and $\omega_x$ is a smoothing $n$-form centered at the point $x$. With his generalization of the smoothing operator, they are able to define a new Lie derivative that takes into account the effect of $\Mc T$ in the variation of $\tilde X$. Then, all the usual features of tensor calculus can be redefined and a Colombeau algebra can be constructed using the new Lie derivative in the definition of moderate and negligible functions. This theory has the advantage of using a smoothing operator which transforms nicely by diffeomorphism, but this is at the price of quite cumbersome new elements in the theory. In particular, the generalized Levi-Civita connection of a generalized metric must keep track of an underlying background connection. Ultimately, their theory is very difficult to use with the calculus of variation, specifically when we shall not assume the existence of an underlying geometry, and when the properties of the manifold are encoded in Euler-Lagrange equations. For this purpose, we take a different path where we do not try to define a diffeomorphic invariant Colombeau algebra, but rather, we express the action of diffeomorphisms as an operator between Colombeau algebras defined on specific coordinate systems.

In the following, we use the following smoothing of tensor fields:
\begin{definition}  \textbf{(coordinate dependent regularization)}
Let ${T^{\mu_1.. \mu_r}}_{\nu_1..\nu_s}$ be the components of a $(r,s)$ tensor field over a manifold $\Mc M$ of dimension $n$. For simplicity, we consider a unique chart. Let us consider a coordinate system denoted by $x=(x^1,..,x^n)$ or by $y=(y^1,..,y^n)$. The coordinate-dependent regularization is defined by:
\begin{equation}
{\tilde T^{\mu_1.. \mu_r}}_{\nu_1..\nu_s} (x)  = \int {T^{\mu_1.. \mu_r}}_{\nu_1..\nu_s} (y) \eta_{\epsilon}(y-x)~ d^n y
\end{equation}
where the shorthand notation $\eta_{\epsilon}(y-x) = \tfrac{1}{\epsilon}\eta(\tfrac{y-x}{\epsilon})$ is used.
\end{definition}
Remark that in this definition, the integral is not performed using a volume-invariant measure. This is because the mapping is coordinate-dependent, and an equivalent definition can be obtained by redefining the mollifier.

We are now interested in the effect of diffeomorphisms on the regularized tensor fields. It is clear that the standard action of a diffeomorphism on $\tilde T$ does not preserve, in general, the regularization operator. In order to preserve the structure, we need to transform both the components of the tensor field and the regularization kernel.

\begin{definition}  \textbf{(Diffeomorphism operator)}
Let a diffeomorphism that transform the coordinates $x$ into $X$ such that $x=f(X)$ and $X=f^{-1}(x)$. The diffeomorphism operator $\hat D(x,X)$ is defined by:
\begin{equation*}
\begin{split}
\left[ \hat D(x,X)  \tilde T\right]^{\mu_1'.. \mu_r'}_{\nu_1'..\nu_s'} =  \int & {\tilde T^{\mu_1.. \mu_r}}_{\nu_1..\nu_s}(x)  \frac{\partial f^{\mu_1'}}{\partial X^{\mu_1}}(f^{-1}(x)) ... \frac{\partial f^{\nu_s}}{\partial X^{\nu_s'}}(f^{-1}(x)) \\ 
& \times \eta'_{\epsilon } \left(f^{-1}(x) - X\right) \frac{1}{\det(\partial f(f^{-1}(x))} d^nx
\end{split}
\end{equation*}
where the mollifier $\eta'$ is not necessarily the same as $\eta$, the one used in the regularization in the coordinate system $x$, and $\partial f$ is a short notation for the Jacobian matrix $\tfrac{\partial f^\mu}{\partial X^{\nu}}$. Note that for simplicity, the same regularization parameter $\epsilon$ is used in both coordinate systems.
\label{def:diffeomorphism_operator}
\end{definition}
\begin{proposition}
Under the assumption that 
\[
\left(\tfrac{\partial f^{\mu_1'}}{\partial X^{\mu_1}}(f^{-1}(x)) ... \tfrac{\partial f^{\nu_s}}{\partial X^{\nu_s'}}(f^{-1}(x))\right)\det(\partial f(f^{-1}(x))^{-1}\]
does not diverge, we have:
\[
\left[ \hat D(x,X)  \tilde T \right]^{\mu_1'.. \mu_r'}_{\nu_1'..\nu_s'} = {\tilde T^{\mu_1'.. \mu_r'}}_{\nu_1'..\nu_s'} (X) + O(\epsilon^{q+1}).
\]
\end{proposition}
%
%
\begin{proof}
This follows from a straightforward calculation already performed in the proof of the Theorem~\ref{thm:scalar_prod_2_Dirac_distrib}. For simplicity of notation, we show the calculations for a vector field $T^\mu$. The case of arbitrary tensor fields follows the same calculation steps.
\begin{align}
 &\int d^nx~ \tilde T^{\mu}(x) \frac{\partial f^{\mu'}}{\partial X^{\mu}}(f^{-1}(x))  \eta'_{\epsilon} \left(f^{-1}(x) - X\right) \frac{1}{\det(\partial f(f^{-1}(x))}  \\
& = \int d^n x ~ d^n y ~ T^\mu(y) \eta_\epsilon(y-x) \frac{\partial f^{\mu'}}{\partial X^{\mu}}(f^{-1}(x))  \eta'_{\epsilon} \left(f^{-1}(x) - X\right) \frac{1}{\det(\partial f(f^{-1}(x))}  \\
& = \int  d^n y \left[ T^\mu(y) \frac{\partial f^{\mu'}}{\partial X^{\mu}}(f^{-1}(y))  \eta'_{\epsilon } \left(f^{-1}(y) - X\right) \frac{1}{\det(\partial f(f^{-1}(y))}  + O(\epsilon^{q+1},y) \right].
\end{align}
In the last line, the variable of the negligible function of order $\epsilon^{q+1}$ is specified in order to specify that it must be integrated once again. Using the fact that the function does not diverge and is of compact support, its integral is still of order $\epsilon^{q+1}$.
Next, we introduce the variable $Y$ such that $y=f(Y)$, and using $d^ny =\det( \partial f) d^nY$, we obtain:
\begin{align}
\left[\hat D(x,X) \tilde T\right]^\mu & =  \int T^\mu(f(Y)) \frac{\partial f^{\mu'}}{\partial X^{\mu}}(Y)  \eta'_{\epsilon } \left(Y - X\right)~  d^n Y + O(\epsilon^{q+1}) \\
& =  \int T^{\mu'}(Y) \eta'_{\epsilon } \left(Y - X\right)~  d^n Y + O(\epsilon^{q+1}) \\
& =\tilde T^{\mu'} (X)+ O(\epsilon^{q+1}).
\label{eq:tilde_T_new_coord}
\end{align}
\end{proof}

As we can see, the diffeomorphism operator modifies the components of the tensor field using the usual rule of tensor transformation, but it also maps the smoothing kernel expressed in one coordinate system into another smoothing kernel expressed in the new coordinate system. Note the introduction of a ramming small "error" of order $o(\epsilon^{q+1})$. This is interpreted by the fact that if we discretize a space in a coordinate system, the discretization may not be accurate when it is mapped into another coordinate system. This is obvious, for example, if we discretize $\setR^2$ with polar coordinates $(r,\theta)$ with constant discretization steps $\Delta r $ and $\Delta \theta$. When $r\rightarrow \infty$, the area of a cell becomes very large, while it is kept constant if we do something similar in Cartesian coordinates. Then, the mapping of discretized functions from one discretization grid to another one may introduce some discretization artifacts.

The increase of the error induced by a change of coordinate can also be viewed by evaluating $\tilde T^{\mu'}(X)$ in Eq.~\eqref{eq:tilde_T_new_coord}:
\begin{equation}
\left[\hat D(x,X) \tilde T\right]^\mu =  T^{\mu'} (X) + O(\epsilon^{q+1},X) +  o(\epsilon^{q+1})
\end{equation}
There are two terms of order $\epsilon^{q+1}$, one is constant, and the other depends on $X$. Hence, in practice when multiple changes of coordinates are performed, one must take care that the accumulation of error does not dominate the factor $\epsilon^{q+1}$. With such a requirement, $o(\epsilon^{q+1},X) +  o(\epsilon^{q+1})$ is a negligible function, both from the point of view of the old and the new coordinate system, and we can make equivalent classes of regularized tensors to create generalized function. The resulting Colombeau algebra is coordinate dependent, but after a change of coordinate, the negligible rest that depends on the old coordinate system is viewed as a negligible constant in the new coordinate system. Then, we can see the diffeomorphism operator as a mapping between Colombeau algebra associated with different coordinate systems.


\section{Action Variation With Generalized Functions}
\label{sec:action_variation_with_GF}

Now that the framework of generalized tensor fields is set, we can revisit the calculus of variation with this formalism. In this section, we present how the variation of a generalized field can be defined and first- and second-order variations of a functional are computed. A specific analysis of the boundary effects is made. This issue is generally not considered with distributions since it leads to ill-defined situations (products involving a singular distribution). As detailed below, with generalized functions, everything is well-defined and can be interpreted in terms of generalized functions.

For the sake of generality, we consider here a functional of the form:
\begin{equation}
I[F] = \int_{\Mc M} F(x)~ d^N x,
\end{equation}
and the specific case of a physical action is considered in the second step.

\subsection{Variations with Respect to a Field}

In this section, we consider that $F$ is an analytic function of a generalized tensor field $\tilde T_\mu$, or a function of the $n$-th derivative of the field $\partial_\nu^{(n)} \tilde T_\mu$. Hence, fixing a component of the field, we have:
\begin{equation}
F= \sum_{k=0}^\infty F_{(k)} \left (\partial_\nu^{(n)} \tilde T_\mu \right)^k ~;~ n \geq 0.
\end{equation}
\textbf{
However, to simplify the equations, the main elements of the theory are given for a scalar field, and the generalization to higher types of tensor field is given in the next step, with specific examples.}

To compute the variation of a functional with respect to a generalized field in a given coordinate system, we need to introduce a perturbation of the field with respect to an initial configuration:
\begin{definition}\textbf{(Field variation)}
Let $\tilde T(x)$ be a generalized vector field over $\Mc M$ expressed in a given coordinate system $x$. The variation $\lambda$ of the generalized scalar field at the position $y$ is defined by:
\[
\delta  \tilde T(x) = \lambda ~ \eta_{\epsilon}(x-y),
\]
\end{definition}
Note that we must have $\lambda\ll \epsilon$ in order to ensure that the variation remains small. Then, in the equations, we must always take the limit $\lambda \rightarrow 0$ before $\epsilon \rightarrow 0$.

\subsubsection{First order variation}
\begin{definition} \textbf{(First order variation of a functional)} The first order variation of a functional $I$ with respect to the variation of a generalized tensor $\tilde T$ field at a given position $y$ is given by the Gateaux derivative~\cite{BSMF_1919__47__70_1,gal2007mathematics,BeranZdenek2013Aofd}:
\begin{equation*}
\begin{split}
\frac{\delta I}{\delta \tilde T(y)} & = \lim_{\lambda \rightarrow 0} \frac{I[\tilde T(x) + \lambda ~ \eta_\epsilon (x-y)] - I[\tilde T(x)]  }{\lambda} \\
&= \left[ \frac{d}{d \lambda} I[\tilde T(x) + \lambda ~ \eta_\epsilon (x-y)] \right]_{\lambda= 0}
\end{split}
\end{equation*}
\label{def:first_order_variation}
\end{definition}

\begin{theorem}
Let $F$ be an analytic function of $\tilde T$ and its $n$-th derivatives. Let the functional $I$ be of the form:
\[
I = \int_{\Mc V} d^Nx F\left(\tilde T(x), \partial_\mu \tilde T (x),...,\partial^{(n)}_\mu \tilde T (x)\right)
\]
with $\Mc V \subseteq \Mc M$ is a compact region of $\Mc M$, or $\Mc V = \Mc M$. Then, the first order variation of $I$ is given by:
\begin{equation}
\begin{split}
\frac{\delta I}{\delta \tilde T (y)} = & \tilde \Id_{\Mc V} (y) \left[ \frac{\partial F}{\partial \tilde T(y)}  +   \sum_{k=1}^n (-1)^k  ~ \partial_\mu^{(k)}\left(\frac{\partial F}{\partial  \left(\partial_\mu ^{(k)} \tilde T(y) \right)} \right) \right]\\ 
&+ \sum_{k=1}^n (-1)^k \left( \partial_\mu^{(k)} \tilde \Id_{\Mc V} (y) \right) \frac{\partial F}{\partial  \left(\partial_\mu ^{(k)} \tilde T(y) \right)} + O(\epsilon^{q+1}),
\end{split}
\label{eq:eq_1_thm_first_order_variation}
\end{equation}
In particular, in the case $n=1$ and $y\not \in \partial V$, we recover the usual Euler-Lagrange equation:
\begin{equation}
\frac{\delta I}{\delta \tilde T (y)} = \frac{\partial F}{\partial \tilde T(y)}  - \partial_\mu\left(\frac{\partial F}{\partial  \left(\partial_\mu \tilde T(y) \right)} \right)  + O(\epsilon^{q+1}).
\label{eq:eq_2_thm_first_order_variation}
\end{equation}
\end{theorem}
\paragraph{Important notes}:
\begin{itemize}
\item Eqs.~\eqref{eq:eq_1_thm_first_order_variation} and \eqref{eq:eq_2_thm_first_order_variation} must be taken in the sense of generalized functions, the $o(\epsilon^{q+1})$ rest is therefore not necessary since we have already considered a class of equivalence. It is kept to emphasize that generalized functions are used.
\item When $y$ is on the boundary of the integration domain, $\tfrac{\delta I}{\delta \tilde T (y)} $ cannot make sense as a tensor field, but only as a generalized tensor field, since $ \partial_\mu^{(k)} \tilde \Id_{\Mc V} (y)$ gives a generalized function which may be associated with a distribution. For example, in the 1D case,
$\partial_x \tilde \Id_{[0,\infty[} (x) \asymp \delta_{0}(x) $.
\end{itemize}

\begin{proof}
As a first step, we regularize the indicator function responsible for the domain of integration $\Mc V$, to ensure that the functional is a well-defined measure:
\[
I = \int_{\Mc M} \tilde \Id_{\Mc V} (x) F(x) ~ d^Nx .
\]
Then, we can easily apply the definition and use the commutation of $d/d\lambda$ with the integral over $\Mc M$:
\begin{equation}
\begin{split}
\frac{d}{d\lambda} I[\tilde T_\mu(x) + \lambda\eta_\epsilon (x-y)] =& \int_{\Mc M} d^Nx ~ \tilde \Id_{\Mc V} (x)\left[ \frac{\partial F}{\partial \left(\tilde T(x) + \lambda\eta_\epsilon (x-y)\right)} \eta_\epsilon (x-y) \right. \\
& + \left.\sum_{k=1}^n \frac{\partial F}{\partial  \left(\partial_\mu ^{(k)} \tilde T(x) + \lambda \partial_\mu^{(k)} \eta_\epsilon (x-y)\right)} \partial_\mu^{(k)} \eta_\epsilon (x-y)\right] 
\end{split}
\label{eq:dI/dlambda}
\end{equation}
We see that we need to evaluate an integral of the form $\int \tilde \Id_{\Mc V} (x) f(x) \eta_\epsilon (x-y) ~ d^Nx$. In the simple case $N=1$, we have:
\[
\int \tilde \Id_{\Mc V} (x) f(x) \eta_\epsilon (x-y) ~ d^Nx = \left\lbrace \begin{array}{ccc}
f(y) + O(\epsilon^{q+1}) & \text{if} & y \not \in \partial \Mc V \\ 
\frac{1}{2}f(y) + O(\epsilon^{q+1}) & \text{if} & y  \in \partial \Mc V
\end{array} \right.
\]
The second case is induced by the fact that when $y$ is on the boundary of $\Mc V$, the function $\eta$ is exactly integrated on the half, and the factor survives in the limit $\epsilon \rightarrow 0$. In higher dimensions, a similar result holds, but the factor depends on the geometry of $\partial \Mc V$ at $y$. For example, if $y$ is a vertex of an $N$-simplex, we have a factor $1/N!$ (see~\cite{ansel2018optimal} Appendix B). In any case, we see that $\braket{\tilde \Id_V (x) f(x)}{\eta_\epsilon(x-y)}=\tilde \Id_V (y) f(y)$

Then, using this fact, and taking $\lambda =0$ in Eq.~\eqref{eq:dI/dlambda}, we deduce:
\begin{equation}
\frac{\delta I}{\delta \tilde T (y)} = \frac{\partial F}{\partial \tilde T(y)} \tilde \Id_{\Mc V}(y) + O(\epsilon^{q+1}) + \sum_{k=1}^n \int d^Nx ~\tilde \Id_{\Mc V} (x)  \frac{\partial F}{\partial  \left(\partial_\mu ^{(k)} \tilde T(x) \right)} \partial_\mu ^{(k)} \eta_\epsilon (x-y).
\end{equation}
The next step is to use the rule of derivatives of distributions to arrive at:
\begin{equation*}
\frac{\delta I}{\delta \tilde T (y)} = \frac{\partial F}{\partial \tilde T(y)} \tilde \Id_{\Mc V}(y)  + \sum_{k=1}^n (-1)^k \int d^Nx ~\eta_\epsilon (x-y) \partial_\mu^{(k)} \left[\tilde \Id_{\Mc V} (x)  \frac{\partial F}{\partial  \left(\partial_\mu ^{(k)} \tilde T(x) \right)} \right] ,
\end{equation*}
which is equal to:
\begin{equation*}
\begin{split}
\frac{\delta I}{\delta \tilde T (y)} = & \tilde \Id_{\Mc V} (y) \left[ \frac{\partial F}{\partial \tilde T(y)}  +   \sum_{k=1}^n (-1)^k  ~ \partial_\mu^{(k)}\left(\frac{\partial F}{\partial  \left(\partial_\mu ^{(k)} \tilde T(y) \right)} \right) \right]\\ 
&+ \sum_{k=1}^n (-1)^k \left( \partial_\mu^{(k)} \tilde \Id_{\Mc V} (y) \right) \frac{\partial F}{\partial  \left(\partial_\mu ^{(k)} \tilde T(y) \right)} + o(\epsilon^{q+1}).
\end{split}
\label{eq:generalized_euler_lagrange_eq_proof}
\end{equation*}
This finishes the first part of the proof.
The demonstration of Euler-Lagrange equation is a straightforward calculation from Eq.~\eqref{eq:generalized_euler_lagrange_eq_proof} once we notice that $ \tilde \Id_{\Mc V} (y) = 1$ and $\partial_\mu^{(k)} \tilde \Id_{\Mc V} (y) = 0$ when $y \not \in \partial \Mc V$.
\end{proof}

\subsubsection{Higher order variations}

\begin{definition} \textbf{(Higher order variation of a functional)}
Following the definition \ref{def:first_order_variation}, we define the $n$-th variation of a functional $I$ by:
\begin{equation*}
\frac{\delta^n I}{\delta \tilde T(y)^n} = \left[ \frac{d^n}{d \lambda^n} I[\tilde T(x) + \lambda ~ \eta_\epsilon (x-y)] \right]_{\lambda= 0}.
\end{equation*}
\label{def:n-th_order_variation_functional}
\end{definition}

In general, we do not have a simple explicit formula like the Euler-Lagrange equation, and the variation can be a singular non-linear generalized function. However, there is a well-known result for the second order of variation that gives in which conditions the extremum of a functional is minimum or a maximum: the Legendre-Clebsh condition~\cite{gelfand2000calculus,bryson1975applied}. This condition is often used in optimal control theory for a sufficient optimality condition. However, when formulated in terms of generalized functions, additional assumptions on the properties of the mollifier are required in the calculations. A few notes on this subject are given in appendix~\ref{sec:Legendre-Clebsh condition}.

\subsubsection{Variation With respect to the element of a group}
\label{sec:variation_group_field}

In this section, we investigate the case when the field must verify a Lie group structure $G$ that can be assimilated to $GL_n(\setR)$, $GL_n(\setC)$, or one of their subgroups. For example, at a given position of the manifold, the field could be an invertible matrix. With the Gateaux derivative presented in Def.~\ref{def:first_order_variation}, the group properties are not preserved, and hence, one has to use a different approach. 

The difficulty lies in the fact that the null matrix is not in $G$, and thus, $g_1(x) + \lambda \eta_{\epsilon}(x) g_2(x) \neq G$ for $g_1,g_2 \in G$. However, the null matrix is in the Lie algebra $\mathfrak{g} = \text{Lie}(G)$. Thus, we can define mollifiers in $\mathfrak{g}$, and make use of the matrix exponential~\cite{hall2000elementary}.

\begin{definition} \textbf{(Generalized Lie-algebra field)} Let $\mathfrak{g}$ be the Lie algebra of a group $G$. Let $\tilde f: \Mc M \mapsto \mathfrak{g}$ be a field over $\Mc M$ such that $\tilde f(x) = \tilde f_A(x) \tau^A$, with $\tilde f^A(x)$ a generalized field over $\Mc M$ and $\tau^A$ the generaors of $\mathfrak{g}$. $\tilde f$ is called a generalized Lie-algebra field.
\end{definition}

\begin{definition} \textbf{(Generalized Lie-group field)} $\tilde f$ be a lie-algebra field. In the case when $G$ is connected, the Lie group field is defined by $\tilde g(x) = \exp(\tilde f(x))$. In the case when $G$ has disconnected parts, there are elements of the group that cannot be generated from the identity, hence we define $\tilde g(x) = h(x)\exp(\tilde X(x))$ with $h:\Mc M \mapsto G$ a gauge field that allows us to generate any element of the group.
\end{definition}

To compute the variations of a generalized Lie-group field, we need to make use of the following formula~\cite{wilcox1967exponential,hall2000elementary}:

\begin{equation}
\frac{d e^{\Mc A}}{d x} = e^{\Mc A} \left( \sum_{n=0}^\infty \frac{(-1)^n}{(n+1)!} \text{ad}_{\Mc A}^n \left(\frac{d \Mc A}{dt}\right)\right) =  e^{\Mc A}  \int_0^1 ds~  e^{-s\Mc A} \frac{d \Mc A}{dx}  e^{s\Mc A}, 
\label{eq:derivative_matrix_exp}
\end{equation}
With $\Mc A$ a matrix function of the parameter $x$ and $\text{ad}_A(B) \equiv AB-BA$. In all generality, the functional variation of $\tilde g$ must be calculated using $\exp\left(f +\lambda \eta_\epsilon \tau^A\right)$, to keep the group properties, and not using $\tilde g + \delta \tilde g$, as it is usually done in the literature. Hopefully, we have the following useful relations

\begin{lemma}{\textbf{First order variations of $e^f$ and its partial derivative}}
\begin{align}
\left[ \frac{d}{d\lambda} e^{f +\lambda \eta_\epsilon \tau^A }\right]_{\lambda =0} & = \eta_\epsilon e^f C^A\\
\left[ \frac{d}{d\lambda} \frac{\partial}{\partial x^\mu} e^{f +\lambda \eta_\epsilon \tau^A}\right]_{\lambda =0} & = \frac{\partial}{\partial x^\mu}  \left(\eta_\epsilon e^f C^A \right) \\
\end{align}
With $C^A =   \sum_{n=0}^\infty \frac{(-1)^n}{(n+1)!} \text{ad}_{f}^n \left(\tau^A\right)$. The tildes have been dropped, to simplify the notation, but also because the formula holds even outside the scope of generalized functions.
\end{lemma}

\begin{proof}
The first equality is proved by a straightforward calculation using Eq.~\eqref{eq:derivative_matrix_exp}.
The second equality is less easy to show. For that purpose, we define $U^s =e^{s(f + \lambda \eta_\epsilon \tau^A)}$, $B(s) = \int_0^s du ~ U^{-u}\left( \partial_\mu f + \lambda \partial_\mu \eta_\epsilon \tau^A \right) U^u$ and $A(s) = \int_0^s du ~ U^{-u}\left(  \eta_\epsilon \tau^A \right) U^u$. All these quantities are assigned to a subscript 0 when the limit $\lambda \rightarrow 0$ is taken, e.g. $U_0^s =e^{s f} $. Using these notations, we have $\partial_\mu U^1 = U^1 B(1)$, $d_\lambda U^1 = U^1 A(1)$, $d_\lambda U^{-s} = - U^{-s}( d_\lambda U^s )U^{-s}$ and $\partial_\mu U^{-s} = - U^{-s}( \partial_\mu U^s )U^{-s}$. Next, we can proceed to an explicit expansion of both sides of our identity:
\begin{align}
\frac{d}{d\lambda} \frac{\partial}{\partial x^\mu} U^1 &= U^1\left\lbrace A(1) B(1) + \frac{dB(1)}{d\lambda} \right\rbrace \\
 \frac{\partial}{\partial x^\mu} \left[ \frac{d}{d\lambda} U^1\right]_{\lambda = 0} &= U^1_0\left\lbrace B_0(1) A_0(1) + \frac{\partial A_0(1)}{\partial x^\mu} \right\rbrace.
\end{align}

An explicit calculation of $d_\lambda B(1)$ and $\partial_\mu A_0(1)$ gives us:
\begin{align}
\frac{dB(1)}{d\lambda} &= \int_0^1 ds \left\lbrace [U^{-s}(\partial_\mu f + \lambda \partial_\mu \eta_\epsilon \tau^A) U^s,A(s)] + U^{-s} \partial_\mu \eta_\epsilon \tau^A U^s \right\rbrace \\
\frac{\partial A_0(1)}{\partial x^\mu} &= \int_0^1 ds \left\lbrace [U_0^{-s}(\eta_\epsilon \tau^A) U_0^s,B_0(s)] + U_0^{-s} \partial_\mu \eta_\epsilon \tau^A U_0^s \right\rbrace.
\end{align}
Note the use of the commutator $[.,.]$ in the last two equations. Now, we have at hand all the ingredients necessary to complete the proof of the second equality, which is verified by computing the difference between the two terms. We have
\begin{equation}
\begin{split}
U^{-1}\left(\left[ \frac{d}{d\lambda}\frac{\partial}{\partial x^\mu}U\right]_{\lambda=0} - \frac{\partial}{\partial x^\mu}\left[ \frac{d}{d\lambda}U\right]_{\lambda=0}\right)= &A_0(1) B_0(1) - B_0(1) A_0(1) \\
&+ \int_0^1 ds\left( [B_0'(s),A_0(s)] -[A_0'(s),B_0(s)]\right),
\end{split}
\end{equation}
With $A_0'(s) = U^{-s}\left(  \eta_\epsilon \tau^A \right) U^s$ and $B_0' = U^{-s}\left( \partial_\mu f \right) U^s$. Next, we can integrate by parts and we are left with:
\begin{equation}
\begin{split}
U^{-1}\left(\left[ \frac{d}{d\lambda}\frac{\partial}{\partial x^\mu}U\right]_{\lambda=0} - \frac{\partial}{\partial x^\mu}\left[ \frac{d}{d\lambda}U\right]_{\lambda=0}\right)= &A_0(1) B_0(1) - B_0(1) A_0(1) \\
&+ B_0(1) A_0(1) - A_0(1)B_0(1) = 0.
\end{split}
\end{equation}
\end{proof}

Variations of an arbitrary Lie-group field are in general, very complicated, however, we can simplify the equations if we impose that the quantity of interest is not the field, but a class of equivalence up to gauge transformation. In such a situation, we can define $\tilde g (x) = h(x) e^{\tilde f(x)}$, with $\Vert \tilde f \Vert \ll 1$. As a consequence, we can look at the variations of $\tilde g$ in the vicinity of the identity, and the term $C^A$ in the equations becomes $C^A = \tau^A$. In this situation, we then have:
\begin{equation}
\frac{\delta \tilde g(x)}{\delta \tilde f_A(y)}  = \eta_\epsilon (x-y) \tilde g(x) \tau^A.
\end{equation}

\subsection{Optimal Control Action}

Now that variations of a functional are defined, we focus on the specific case of the optimal control action~\cite{bryson1975applied,contreras_dynamic_2017,ito2008lagrange}. In this section, we determine the first- and second-order variations of this latter.

 \textit{To simplify the notation, we drop the tilde above the generalized function unless it is necessary.}

We consider an optimal control problem defined by a system $q(t) \in \Gamma \simeq \setR^n$ (the phase space or the configuration space), $u(t)\in U \subseteq \setR^m$ is a control field, and $p(t) \in \Gamma$ an adjoint state that allows us to measure the deviation of the system from the trajectory of $q$ given by the equation of motion. We introduce the extended configuration space $(q(t),p(t),u(t)) \in \Mc C = \Gamma \times \Gamma \times U$ and we define the following action to minimize:
\begin{equation}
S = h_f\left( q(t_f)\right) -  h_i\left( q(t_i) \right) + \int_{t_i}^{t_f} dt~\left(p_\mu(t) \left[ \frac{d q^\mu}{dt}(t) - f^\mu (q,u,t)  \right] + f^{(0)}(q,u,t) \right)
\label{eq:action_oc}
\end{equation}
We have introduced $h_f$ and $h_0$ initial and terminal costs, and $ f^{(0)}$ a quantity to minimize overtime called dynamical cost. $f^\mu$ is a function that defines the equation of motions of $q^\mu$, given by $d_t q^\mu = f^\mu$. Note that $p$ is a true adjoint state of $q$ since $p_\mu = \tfrac{\partial L}{\partial (d_t q^\mu)}$.
For simplicity, we make the calculation with the standard case of an optimal control system where the dynamical cost is an integral over time, but we can generalize the calculation with an integral over space-time. In a standard optimal control problem, we usually have $h_i = 0$ and $h_f$ is a measure of distance between $q(t_f)$ and a target state $q_{target}$. Also, we usually have $f^{(0)}=1$ if we want to reach the target state in minimum time, or $f^{(0)} = u^2/2$ if we want to control with a minimum of energy~\cite{bonnard_optimal_2012,bryson1975applied}.

We shall stress that the optimal control action is, strictly speaking, not new in physics, in the sense that it is nothing else than a functional that imposes constraints using a Lagrange multipliers~\cite{Courant_Hilber_vol_1,Courant_Hilber_vol_2,gelfand2000calculus,bryson1975applied}. This is as old as the Hamiltonian formalism since it is with this underlying idea that we can pass from the Lagrangian formalism into the Hamiltonian one. However, what is different is its systematic use for the definition of the system action~\cite{bryson1975applied,contreras_dynamic_2017}. The optimal control Lagrangian can be seen as an intermediate formalism between the usual Lagrangian and the usual Hamiltonian of the system since we work in an extended configuration space, which looks like the phase space of a Hamiltonian system. The advantage of the approach is that an optimal control action can be defined for any time continuous dynamical system, while a "standard" Lagrangian cannot be constructed for any arbitrary dynamical system~\cite{contreras_dynamic_2017}. The physical interpretation of the adjoin state is given in Sec.~\ref{sec:physical interpretation}.

\subsubsection{First and Second Order Variations of the Optimal Control Action}
\paragraph{Firdt order variation}

Using the definition \ref{def:first_order_variation}, we can compute the first-order variations of the functional and the variations of the boundary costs $h_i$ and $h_f$. We have explicitly:
\begin{align}
\label{eq:dS/dq_oc_action}
\frac{\delta S}{\delta q^\mu(\tau)}=& \frac{\partial h_f}{\partial q^\mu(t_f)}  \eta_\epsilon (\tau - t_f) - \frac{\partial h_i}{\partial q^\mu(t_i)}\eta_\epsilon (\tau - t_i) + p_\mu(\tau) \left[ \eta_\epsilon (\tau - t_f) - \eta_\epsilon(\tau -t_i) \right] \\
  & + \left[ \frac{\partial f^{(0)}}{\partial q^\mu(\tau)} - p_\nu(\tau) \frac{\partial f^\nu}{\partial q^\mu(\tau)}   - \frac{dp_\mu}{d\tau} \right] \Id_{[t_i,t_f]}(\tau) + O(\epsilon^{q+1} ),\notag \\
\label{eq:dS/dp_oc_action}
 \frac{\delta S}{\delta p_\mu(\tau)} =& \left[ \frac{d q^\mu}{dt}(t) - f^\mu (q,u,t) \right]  \Id_{[t_i,t_f]}(\tau) + O(\epsilon^{q+1} ), \\
 \label{eq:dS/du_oc_action}
\frac{\delta S}{\delta u^i(\tau)} =&  \left[ \frac{\partial f^{(0)}}{\partial u^i(\tau)} - p_\nu(\tau) \frac{\partial f^\nu}{\partial u^i(\tau)} \right] \Id_{[t_i,t_f]}(\tau) + O(\epsilon^{q+1} ).
\end{align}
With these three equations, we recover straightforwardly all the ingredients of the well-known weak Pontryagin Maximum Principle, using $\delta S = 0$.  Eq.~\eqref{eq:dS/dp_oc_action} gives us the equation of motion of the system: $d_t q^\mu = f^\mu$, Eq.~\eqref{eq:dS/du_oc_action} is the equivalent of $\tfrac{\partial H_P}{\partial u^i} = 0$ (where $H_p$ is Pontryagin's Hamiltonian), and Eq.~\eqref{eq:dS/dq_oc_action} gives both the equations of motion of the adjoint state $d_t p_\mu =  \tfrac{\partial f^{(0)}}{\partial q^\mu} - p_\nu \tfrac{\partial f^\nu}{\partial q^\mu}$, and the boundary conditions of the adjoint state: $p_\mu(t_i) = - \tfrac{\partial h_i}{\partial q^\mu(t_i)}$ and $p(t_f) = - \tfrac{\partial h_f}{\partial q^\mu(t_f)}$. Solutions to this system of equations are possible optimal solutions to the control problem. Several extremums may exist and a nontrivial task is to select which one is the global one. We refer to \cite{bonnard_optimal_2012,bryson1975applied} for additional information on the topic.

\paragraph{Second order variations}

\begin{equation}
\begin{split}
\frac{\delta^2 S}{\delta q^\mu(t_1)\delta q^\nu(t_2)} = & \frac{d^2 h_f}{dq^\mu dq^\nu} \eta_\epsilon(t_1-t_f)\eta_\epsilon(t_2-t_f) - \frac{d^2 h_i}{dq^\mu dq^\nu} \eta_\epsilon(t_1-t_i)\eta_\epsilon(t_2-t_i) \\
&+\left[ \frac{\partial^2 f^{(0)}}{\partial q^\mu \partial q^\nu}(t_1) - p_\sigma (t_1) \frac{\partial^2 f^{\sigma}}{\partial q^\mu \partial q^\nu}(t_1) \right]  \Id_{[t_i,t_f]}(t_1) \eta_\epsilon (t_1-t_2) + O(\epsilon^{q+1})
\end{split}
\end{equation}
\begin{align}
\frac{\delta^2 S}{\delta u^i(t_1)\delta q^\mu(t_2)}& = \left[ \frac{\partial^2 f^{(0)}}{\partial u^i \partial q^\mu}(t_1) - p_\sigma (t_1) \frac{\partial^2 f^{\sigma}}{\partial u^i \partial q^\mu}(t_1) \right]  \Id_{[t_i,t_f]}(t_1) \eta_\epsilon (t_1-t_2) + O(\epsilon^{q+1})\\
\frac{\delta^2 S}{\delta u^i(t_1)\delta p_\mu(t_2)}& =  - \frac{\partial f^{\mu}}{\partial u^i (t_1)} \Id_{[t_i,t_f]}(t_1) \eta_\epsilon (t_1-t_2) + O(\epsilon^{q+1}) \\
\frac{\delta^2 S}{\delta u^i(t_1)\delta u^j(t_2)}& = \left[ \frac{\partial^2 f^{(0)}}{\partial u^i \partial u^j}(t_1) - p_\sigma (t_1) \frac{\partial^2 f^{\sigma}}{\partial u^i \partial u^j}(t_1) \right] \Id_{[t_i,t_f]}(t_1) \eta_\epsilon (t_1-t_2) + O(\epsilon^{q+1}) \\
\frac{\delta^2 S}{\delta p_\mu(t_1)\delta q^\nu(t_2)}& = - \frac{\partial f^{\mu}}{\partial q^\nu (t_1)} \Id_{[t_i,t_f]}(t_1) \eta_\epsilon (t_1-t_2) + O(\epsilon^{q+1}) \\
\frac{\delta^2 S}{\delta p_\mu(t_1)\delta p_\nu(t_2)}& = 0 \\
\end{align}

The second-order variation is generally useful to determine if the extremum is a maximum or a minimum, by analyzing the sign of $\delta^2S$. Here, we see that the second order variation is a generalized function that depends on $\eta_\epsilon(t_1-t_2)$. As a consequence, the sign of the second-order variation depends on the mollifier.

\section{Application to Physical Systems}
\label{sec:anlysis_action_physical_syst}

In this section, we apply the theory to several physical systems. The analysis focuses on the optimal control action since it plays a crucial role in the path integral considered later in this paper. However, other types of actions are also considered, and they are compared together. We start with the harmonic oscillator and next, we address the more complicated case of the scalar field. Finally, we explore the case of Lorentzian manifolds, which plays a central role in quantum gravity. Note that in the following, any fields are given by generalized fields, but the tilde over the letters is generally omitted.

\subsection{Harmonic Oscillator}
\label{sec:harmonic_oscillator}

For this study of the harmonic oscillator~\cite{messiah1962quantum,zinn2004path,gardiner2004quantum}, we compare the standard quadratic action, the optimal control action, and a third one similar to the optimal control action, but which contain a second-order differential equation. The first study of the optimal control action has been addressed by the author in~\cite{ansel2022loop}. A deeper analysis is provided in this paper.

\subsubsection{Standard action}

Let $q$ be the position of the oscillator of mass $m$ and spring constant $k$. The (standard) quadratic action defining the equation of motions is~\cite{zinn2004path}:
\begin{equation}
S_{quad} = \int_{t_i}^{t_f} dt\left(\frac{m}{2} \left(\frac{dq}{dt} \right)^2 - \frac{k}{2} q(t)^2 \right)
\label{eq:S_quad_ho}
\end{equation}
variations with respect to $q$ gives us:
\begin{equation}
\frac{\delta S_{quad}}{\delta q(t)} = \Id_{[t_i,t_f]}(t) \left( - k q(t) -m \frac{d^2q}{dt^2}\right) - m \frac{dq}{dt} (\eta_{\epsilon}(t-t_i) - \eta_{\epsilon}(t-t_f)) + O(\epsilon^{q+1})
\end{equation}
Using $\delta S = 0$, we recover the usual equation of motion with a second-order derivative. We also have a distributional term at the boundary. A boundary term is thus necessary to regularize the final result. This is something usually absent in most of the discussions on this system. The effect of the boundary is illustrated by numerical computations in Sec.~\ref{sec:Numerical comparison between the actions}. Since this action is one of the most studied in the literature, we go next to the presentation of the other actions.

\subsubsection{Optimal control action}
\label{sec:OC_action_ho}

For the optimal control action, we need to rewrite the equation of motion with a first-order differential system, i.e., they must be put in Hamiltonian form. Then, we introduce $p= m \tfrac{dq}{dt}$, and to make an explicit link with the quantum theory, we introduce a complex-valued variable~\cite{messiah1962quantum,zinn2004path,gardiner2004quantum} $\alpha = \tfrac{\ii}{\sqrt{2 \omega m}} p + \sqrt{\tfrac{k}{2 \omega}} q$, where $\omega = \sqrt{k/m}$ is the oscillator frequency. With this complex variable, the equation of motion is given by $d_t\alpha = -\ii \omega ~\alpha$. A quantum harmonic-oscillator has coherent states which satisfy the following scalar product: $\braket{\alpha}{\beta} = \exp(-\vert\alpha\vert^2/2 -\vert\beta\vert^2/2 + \alpha^* \beta)$~\cite{gardiner2004quantum}. This scalar product can be used as a terminal cost in the optimal control action. Taking all these features into account, we define the following optimal control action:
\begin{equation}
S_{OC} = -\ii \log \left( \braket{\beta_f}{\alpha(t_f)} \right) + \ii \log \left( \braket{\beta_i}{\alpha(t_i)} \right)+ \Re \left[ \int_{t_i}^{t_f} dt ~ \pi(t) \left( \frac{d \alpha}{dt} + \ii \omega\alpha (t)\right) \right]
\label{eq:action_oc_ho}
\end{equation}
The configuration space of the OC system is $\Mc C = \setC ^2 = \{ (\alpha , \pi)\}$, with $\pi$ the adjoint state of $\alpha$. Note that here, we do not have a control field, but we can include one without difficulty.
Equations of motion of the extended system are given by $\delta S = 0$. Trivially, we obtain : $d_t \alpha(t) = - \ii \omega \alpha(t) \Id_{[t_i,t_f]}(t)$, and the condition on the adjoint state is:
\begin{equation}
\begin{split}
& \ii \alpha(t_f)  \eta_\epsilon (t - t_f) - \ii \alpha(t_i) \eta_\epsilon (t - t_i) + \pi^*(t) \left[ \eta_\epsilon (t - t_f) - \eta_\epsilon(t -t_i) \right] \\
& - \left[  \frac{d \pi^*}{dt} + \ii \omega~\pi^*(t) \right] \Id_{[t_i,t_f]}(t) + O(\epsilon^{q+1} )  = 0.
\end{split}
\label{eq:dS/dq_oc_action}
\end{equation}
Then, in the limit $\epsilon \rightarrow 0$, we must have:
\begin{align}
\frac{d \pi^*}{dt} &= -  \ii \omega~\pi^*(t) ~;~ t \in [t_i,t_f] \\
\pi^*(t) &= \ii \alpha (t)  ~;~ t  = t_i \text{ or } t_f.
\label{eq:adjoint_state_constrains_ho}
\end{align}
The boundary conditions forces that when $\delta S=0$, we have $\pi = \ii \alpha^*$. Inserting this condition into Eq.~\eqref{eq:action_oc_ho}, we obtain the following dynamical part of the action:
\begin{equation}
S_{holo}=- \Im \left[ \int_{t_i}^{t_f} dt ~\left( \alpha^* \frac{d \alpha}{dt} + \ii \omega \vert \alpha \vert^2 \right) \right].
\label{eq:action_holo_ho}
\end{equation}
This is the usual action of a harmonic oscillator in the holomorphic representation~\cite{zinn2004path}. With this computation, we see that the optimal control action is more than an alternative way to derive the classical equation of motions, it is a more general description of the system for which the standard case is recovered when specific boundary conditions are specified.

\subsubsection{Action with second-order derivatives}

Following the idea of the optimal control action, we can also take the equation of motion obtained with $S_{quad}$, and multiply it by a Lagrange multiplier $\pi$ to form a Lagrangian. Without terminal cost, this would lead to:
\begin{equation}
S_{2nd} = \int_{t_i}^{t_f} dt~ \pi(t) \left( k q(t)+ m \frac{d^2q}{dt^2}\right).
\label{eq:action_second_order_ho}
\end{equation}
Note that contrary to Eq.~\eqref{eq:action_oc_ho}, $\pi$ is now a real variable. The first order variations with respect to $q$ and $\pi$ give us:
\begin{align}
\frac{\delta S_{2nd}}{\delta \pi(t)} &= \Id_{[t_i,t_f]}(t) \left( k q(t)+ m \frac{d^2q}{dt^2}\right) \\
\frac{\delta S_{2nd}}{\delta q(t)} &= k \Id_{[t_i,t_f]}(t)  + 2 m \frac{d \pi}{dt} \frac{d \Id_{[t_i,t_f]}}{dt}+ m \pi(t) \frac{d^2 \Id_{[t_i,t_f]}}{dt^2}
\end{align}
We recover the equations of motion, but we also have a distribution at $t_i$ and $t_f$, given by $d^2_t \Id_{[t_i,t_f]}$, which is equivalent to the derivative of a Dirac distribution. Then, it is not possible to eliminate this variation in a similar way as in Sec.~\ref{sec:OC_action_ho}, unless the terminal cost is of the form $h(t_i) = -m Q(t_i) d_t q(t_i) - 2 d_t Q(t_i) q(t_i) $. This kind of boundary term allows us to specify boundary conditions for the adjoint state, but its relation to a well-defined scalar product is not elucidated yet.

\subsubsection{Numerical comparison between the actions}
\label{sec:Numerical comparison between the actions}

In this section, we illustrate numerically the differences between several actions of the harmonic oscillator near an extremum. In Fig.~\ref{fig:traj_ho}, we show a trajectory of the system computed by extremizing the actions $S_{quad}$, $S_{OC}$, and $S_{holo}$. No boundary terms $h_i$ or $h_f$ are taken into account in the actions, to provide a fair comparison between the different models, but the boundary of the interval of integration is kept. We observe that with the quadratic action, the boundary of the integration domain leads to divergences of $q(t)$, while in the other models, the boundary is responsible for the damping of some of the quantities. Moreover, no divergences are observed with $S_{OC}$ and $S_{holo}$.

The behaviors of the action near the extremal trajectories presented in Fig.~\ref{fig:traj_ho} are explored by computing $S$, $\delta S$, and $\delta^2S$ as a function of the value of $\epsilon$. Results are given in Tab.~\ref{tab:action_oc}. We observe that the values of the action associated with the same trajectory are different. The values depend on the value of $\epsilon$, used to regularize the integration domain, but numerical computations suggest that they converge towards a specific value. As expected, $\delta S \approx 0$ in all cases, since the computations are performed at an extremal, but the value of $\delta^2 S$ depends on the model. Both for $S_{quad}$ and $S_{holo}$ we have a divergence when $\epsilon \rightarrow 0$, but it is identically $0$ with $S_{OC}$.

All these actions have the same predictive power in terms of modeling classical physical systems, however, without terminal cost, they are not strictly identical to each other. The equivalence in the quantum theory can be made with specific processing at the boundary of the integration domain. For the quadratic action, we have to use a way that removes the divergences in the calculation, while for the OC action, the terminal cost must encode the scalar product between two quantum states. The quantum theory for the harmonic oscillator with $S_{OC}$ is considered in Sec.~ \ref{sec:opt_path_integral_ho} and the standard path integral with $S_{quad}$ is revisited with generalized functions in appendix~\ref{sec:appendix_path_integral_ho}.
\begin{figure}
\includegraphics[width=\textwidth]{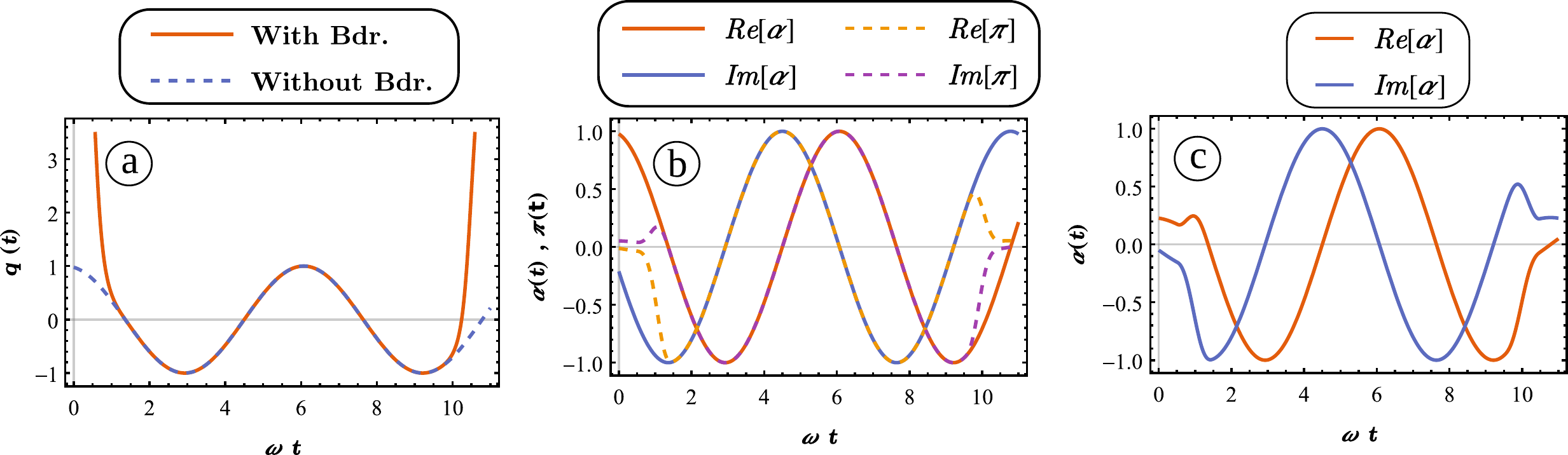}
\caption{Classical trajectory of the Harmonic Oscillator computed using different actions. a) Standard quadratic action $S_{quad}$ defined in Eq.~\eqref{eq:S_quad_ho}, b) Optimal control action $S_{OC}$ defined in Eq.~\eqref{eq:action_oc_ho} without boundary terms and c) action with holomorphic variables $S_{holo}$ defined in Eq.~\eqref{eq:action_holo_ho}. No boundary terms are taken into account in order to provide a fair comparison between the different models. The domain of integration $[t_i,t_f]$ is given by $t_i = 1/\omega$, $t_f = 10/\omega$. To integrate the dynamics, the initial condition is taken at $t=4.5/\omega$ such that $\alpha = \ii$ and $\pi = 1$. In the case of panel a), the initial condition is $q=0$ and $\tfrac{dq}{dt} = 1$. Extremals of the action are computed using regularized function. In order to keep a maximum of analytic computations the mollifier $\eta_\epsilon(x) =  \cos(\pi x/(2 \epsilon)^2/\epsilon \Id_{[-\epsilon,\epsilon]}(x) $ is used. Similar results can be obtained with a true $C^\infty$ function. Here, the value $\epsilon = 0.5$ is used.}
\label{fig:traj_ho}
\end{figure}

\begin{table}[]
\resizebox{\textwidth}{!}{%
\begin{tabular}{c|ccc|ccc|ccc|}
\cline{2-10}
 &
  \multicolumn{3}{c|}{$S_{quad}$} &
  \multicolumn{3}{c|}{$S_{OC}$} &
  \multicolumn{3}{c|}{$S_{holo}$} \\ \hline
\multicolumn{1}{|c|}{$\epsilon$} &
  \multicolumn{1}{c|}{$S$} &
  \multicolumn{1}{c|}{$\delta S$} &
  $\delta^2 S$ &
  \multicolumn{1}{c|}{$S$} &
  \multicolumn{1}{c|}{$\delta S$} &
  $\delta^2 S$ &
  \multicolumn{1}{c|}{$S$} &
  \multicolumn{1}{c|}{$\delta S$} &
  $\delta^2 S$ \\ \hline
\multicolumn{1}{|c|}{1} &
  \multicolumn{1}{c|}{3.96} &
  \multicolumn{1}{c|}{1.33 $\times 10^{-9}$} &
  2.47 &
  \multicolumn{1}{c|}{3.06 $\times 10^{-8}$} &
  \multicolumn{1}{c|}{7.61 $\times 10^{-17}$} &
  0 &
  \multicolumn{1}{c|}{-8.78} &
  \multicolumn{1}{c|}{4.43 $\times 10^{-10}$} &
  -0.75 \\ \hline
\multicolumn{1}{|c|}{0.5} &
  \multicolumn{1}{c|}{3.31} &
  \multicolumn{1}{c|}{2.99 $\times 10^{-9}$} &
  19.74 &
  \multicolumn{1}{c|}{5.33 $\times 10^{-8}$} &
  \multicolumn{1}{c|}{-2.06 $\times 10^{-16}$} &
  0 &
  \multicolumn{1}{c|}{-9.00} &
  \multicolumn{1}{c|}{2.27 $\times 10^{-10}$} &
  -1.5 \\ \hline
\multicolumn{1}{|c|}{0.1} &
  \multicolumn{1}{c|}{2.45} &
  \multicolumn{1}{c|}{1.64 $\times 10^{-8}$} &
  2467.4 &
  \multicolumn{1}{c|}{4.17 $\times 10^{-8}$} &
  \multicolumn{1}{c|}{9.44 $\times 10^{-16}$} &
  0 &
  \multicolumn{1}{c|}{-9.14} &
  \multicolumn{1}{c|}{4.29 $\times 10^{-11}$} &
  -7.5 \\ \hline
\multicolumn{1}{|c|}{0.01} &
  \multicolumn{1}{c|}{2.23} &
  \multicolumn{1}{c|}{1.57 $\times 10^{-7}$} &
  2.47 $\times 10^6$ &
  \multicolumn{1}{c|}{2.97 $\times 10^{-8}$} &
  \multicolumn{1}{c|}{-8.16 $\times 10^{-14}$} &
  0 &
  \multicolumn{1}{c|}{-9.17} &
  \multicolumn{1}{c|}{-9.15 $\times 10^{-12}$} &
  -75 \\ \hline
\end{tabular}%
}
\caption{Value of the actions $S_{quad}, S_{OC}, S_{holo}$ and their first and second order variations near the extremal trajectories of Fig.~\ref{fig:traj_ho} for different values of $\epsilon$. The variation is centered at $t=4.5/\omega$. We refer to the caption of  Fig.~\ref{fig:traj_ho} for the parameters used in the computation of the extremal. The variations are computed using the same mollifier as the one used to regularize the functions.}
\label{tab:action_oc}
\end{table}

\subsection{The Scalar Field}

In this section, we consider the case of the real scalar field~\cite{itzykson2012quantum} in flat space $\psi:\Mc M \rightarrow \setR$. This field satisfies the wave equation:
\begin{equation}
\left(\partial^\mu \partial_\mu - m^2\right) \psi(x) =0.
\end{equation}
Similarly to the harmonic oscillator (see Sec.~\ref{sec:harmonic_oscillator}), different kinds of actions can be defined for which the wave equation is recovered with a first-order variation. Without boundary terms, the equivalent of the three actions studied for the harmonic oscillator is:
\begin{align}
S_{quad} &= \frac{1}{2} \int_V d^4x( \partial^\mu \psi \partial_\mu \psi - m^2 \psi^2) \\
S_{OC}	&= \int_V d^4 x \left( \phi(x) [\partial^\mu P_\mu(x) - m^2 \psi(x)] + \pi^\mu (x) [\partial_\mu \psi (x) - P_\mu (x)]\right) \\
S_{2nd} &=   \int_V d^4x ~\Phi(x) \left(\partial^\mu \partial_\mu - m^2\right) \psi(x)
\end{align}
where $P_\mu$ is a field which becomes $P_\mu(x) = \partial_\mu \phi (x)$ at an extremal, and $\phi,\pi,\Phi$ are adjoint states. Straightforward calculations give us a series of equations very similar to the harmonic oscillator in sec.~\ref{sec:harmonic_oscillator} (explicit calculations are left for the reader).

The first-order variations of these actions suggest different definitions of boundary terms. We focus on a suitable boundary term for the optimal control action. Ultimately, we would like this term to have the same properties as the one of the harmonic oscillator, studied in Sec.~\ref{sec:harmonic_oscillator}. However, we have to overcome a few subtleties, since the boundary of $V$ is not necessarily time oriented.

Here, we propose the following construction. We introduce $\chi(x) = \psi(x) + \ii n^\mu(x) P_\mu(x)$, with $n^\mu(x)$ a normal vector of $\partial V$ at the point $x$ (the orientation plays a role in the path integral but we left this point aside for the moment). We define the boundary term as follows:
\begin{equation}
h = \ii \log \braket{\zeta}{\chi}
\label{eq:bdr_scalar_field}
\end{equation}
with
\begin{equation}
\braket{\zeta}{\chi} = \exp\left(
\int_{\partial V} d^3x \left\lbrace \frac{1}{2} |\zeta(x) - \chi(x)|^2  + \ii \Im(\zeta^*(x)\chi(x))\right\rbrace
\right).
\label{eq:scalar_product_coherent_states_for_scalar_field}
\end{equation}
We can verify that this choice of boundary cost function forces the adjoint state to take the values of the system state on $\partial V$. This is easily seen by the fact that the real part of $h$ is of the form:
\begin{equation}
\Re h = \int_{\partial V} d^3x ~ n^\mu\left(\psi(x) P_\mu'(x) - \psi'(x) P_(x) \right)
\label{eq:antisymmetric_form_scalar_field}
\end{equation}
where primed variables are referred to $\zeta$ and unprimed variables to $\chi$. The first order variation of the total action gives us $\pi^\mu n_\mu = P_\mu 'n^\mu$, $\phi = \psi'$, and $\zeta = \chi$. Hence we have $\pi^\mu n_\mu = P_\mu n^\mu$, $\phi = \psi$. Note that a similar mechanism is also present for the harmonic oscillator (see Eq.~\eqref{eq:adjoint_state_constrains_ho}), but it is not explicit since we have used complex differentiation in the equations. We shall also underline that the antisymmetric form given in Eq.~\ref{eq:antisymmetric_form_scalar_field} is at the core of the construction of Hadamard states~\cite{gerard2018introduction,khavkine2015algebraic}, which can be viewed as a generalization of the vacuum state in QFT in curved space-time.

Eq.~\eqref{eq:scalar_product_coherent_states_for_scalar_field} is equivalent to the standard scalar product of coherent states when the boundary is spacelike. To see this point, we recall that the scalar product between two (usual) coherent states $\eta_1$, $\eta_2$ of the scalar field, given in the momentum basis is~\cite{itzykson2012quantum}
\begin{equation}
\braket{\eta_1}{\eta_2} = \exp\left(
\int_{\partial V} \frac{d^3k}{(2\pi)^2} \frac{1}{2\omega(k)} \left\lbrace \frac{1}{2} |\eta_1(k) - \eta_2(k)|^2  + \Im(\eta_1^*(x)\eta_2(x))\right\rbrace
\right).
\label{eq:scalar_product_coherent_states_for_scalar_field}
\end{equation}
with $\omega(k) = \sqrt{k_a k^a + m^2}$. We can recover Eq.~\eqref{eq:scalar_product_coherent_states_for_scalar_field} directly if we define:
\begin{equation}
\eta_2(k) = \sqrt{2 \omega(k)} \int d^3x \left( \psi(x) + \ii P_0(x) \right)e^{\ii k_a x^a},
\end{equation}
and similarly for $\eta_1$, but with $\psi'$ and $P_0'$.

We finish this section with a comment on a possible boundary cost for $S_{2nd}$. It shall be of the form: $h=-\int_{\partial V}d^3x n^\mu\left( \phi(x) \partial_\mu \psi(x) + 2 \psi(x) \partial_\mu \phi(x) \right)$ where $\phi$ is another scalar field. Once again, it is not clear yet if this boundary term can be used for a quantum scalar product.

\subsection{Lorentzian Space-time}
\label{sec:lorentzian_space_time}

In this section, we investigate how an optimal control action can be defined to characterize a Lorentzian manifold. The issue has been already considered by the author in~\cite{ansel2022loop}, using the standard variables used in loop quantum gravity. Here, a similar approach is used, but using variables closer to the ADM framework~\cite {misner_gravitation_1973}, for which the physical interpretation is easier. Contrary to the previous paper, a specific focus is made on the adjoint state and the boundary term.

The manifold $\Mc M$ is described using a tetrad field $e$. For simplicity, we consider a single-coordinate chart for $\Mc M$.

with an action principle, the data of the geometry must be determined from first principles derived from the extremals of an action. This way, we must assume a minimum mathematical structure and the manifold structure must be encoded in an action. 

\begin{definition} \textbf{(tetrad field)}~\cite{rovelli_quantum_2007,rovelli_covariant_2014,hamilton2018general}
\begin{equation}
\begin{array}{cccc}
e : & \setR^4 & \mapsto     & GL_4(\setR) \\
    & x       & \rightarrow & e_\mu^I (x)
\end{array}
\end{equation}
\end{definition}
The tetrad field allows us to reconstruct the metric of the manifold through the identity:
\begin{equation}
g_{\mu\nu}(x) = e^I_\mu (x)~ e^J_\nu (x) ~\eta_{IJ},
\end{equation}
where $\eta_{IJ} = \text{diag}(-1,1,1,1)_{IJ}$ is the Minkowski metric. The tetrad indices $I,J,K,...$ are moved up and down using the Minkowski metric $\eta$ while coordinates indices $\mu,\nu,\rho..$ are moved using the metric tensor $g$. Moreover, a tensor field can be transformed into a tetrad tensor with contractions with the tetrad. For example $v^\mu e_\mu^I = v^I$, and using the identity $e^I_\mu e^\mu_J = \delta^I_J$, we have $v^I e_I^\mu = v^\mu$. These relations are easily generalized to tensors and cotensors of arbitrary type~\cite{rovelli_quantum_2007,rovelli_covariant_2014,hamilton2018general}.
%

The geometry of the space-time manifold is determined using Einstein's field equations (EFE)
\begin{equation}
G_{\mu \nu}[e](x) = \kappa~ T_{\mu \nu}(x)
\end{equation}
where $G$ is the Einstein tensor, $T$ is the stress-energy tensor and $\kappa$ is a coupling constant.
Einstein's tensor is defined by:
\begin{equation}
G_{\mu \nu}[e] = R_{\mu \nu}[e] + g_{\mu \nu}[e] \left( \Lambda - \frac{1}{2} R[e]\right)
\label{eq:EFE_v1}
\end{equation}
with $\Lambda$ the cosmological constant, $R_{\mu \nu}$ the Ricci curvature tensor, and $R$ the scalar curvature. 
In the form of Eq.~\eqref{eq:EFE_v1}, we cannot directly define an optimal control action with good boundary properties. If we choose a Lagrangian of the form $P^{\mu \nu}(G_{\mu\nu} - \kappa T_{\mu\nu})$, we have variations similar to the derivative of a Dirac distribution at the boundary, exactly like in the case of the second order action \eqref{eq:action_second_order_ho} of the harmonic oscillator. For a similar reason, it is also not clear how the boundary term can be assimilated with a well-defined scalar product in the space of the solution of Einstein's field equation. This is a point that would be interesting to clarify in the future since it may give us a simpler connection between canonical quantum gravity, and the standard formulation of general relativity. In the meantime, we avoid this difficulty and we use a 3+1 formalism where the propagation of the field is studied along a given time-like coordinate. With such an approach, the covariance is less obvious, but the connection with the vast literature of quantum gravity is easier (although it remains non-trivial).

With the 3+1D decomposition, the tetrad becomes~\cite{pons_gauge_2000}:
\begin{equation}
e_{\mu}^I = \left( \begin{array}{cc}
\mathsf{N} & 0 \\ 
e^i_a \mathsf{N}^a &  e_a^i
\end{array} \right),
\end{equation}
with $\mathsf{N}$, $\mathsf{N}^a$ the lapse function, and the shift vectors, which are gauge degrees of freedom. $e^a_i$ is called the triad, and it verifies similar group properties as the tetrad, but the group is now $GL_3(\setR)$ instead of $GL_4(\setR)$. Following the discussion of Sec.~\ref{sec:variation_group_field}, we express $e$ in terms of a generalized field defined on $\mathfrak{gl}_3(\setR)$:
\begin{equation}
e^i_a(x) = \exp\left(f(x)\right)^i_a = \exp \left( f_A(x) \tau^A\right)^i_a.
\label{eq:link_f_to_e}
\end{equation}
In the rightmost term of the equality, the generators $\tau^A$ of $\mathfrak{gl}_3(\setR)$ have been introduced.

We shall note that with this formula, any elements of the group cannot be generated by this identity. This is because $GL_3(\setR)$ is not simply connected. To see this point, consider the singular value decomposition $e = u.d.v^t$, where $t$ denotes the matrix transpose. Here, $u,v \in O(3)$ and $d$ is an arbitrary real diagonal matrix. From Eq.~\eqref{eq:link_f_to_e}, we expect that $d$ is of the form $d= \exp(d')$ with $d'$ another real diagonal matrix. It is then obvious that $d$ is positive and we cannot have a negative eigenvalue, as it can be expected for an arbitrary element of $GL_3(\setR)$. Then, with Eq.~\eqref{eq:link_f_to_e}, only a subset of $GL_3(\setR)$ closed to the identity can be generated. Other subsets of the group can be obtained by matrix multiplication with another (well-chosen) element of the group. This constitutes a change of frame induced by a change of coordinate system (we recall that the tetrad is modified according to $e^I_{\mu'} = e^I_\mu \tfrac{dx^\mu}{dx^{\mu'}}$). The change of the coordinate system is not related to the manifold geometry and it is an input of the theory, to define the variational principle.

We also introduce the extrinsic curvature~\cite{arnowitt_dynamics_1962,friedrich_cauchy_2000}:
\begin{equation}
k^i_a = \frac{e^{bi} }{2\mathsf{N}} \left( \partial_t (e_a^j e_b^k)\delta_{jk} + D_{(a}\mathsf{N}_{b)} \right),
\end{equation}
where the operator $D_a$ is the covariant derivative of the 3-metric, and $t$ is the time parameter. Note that $k^i_a$ can be zero $\forall i,a$, e.g. for flat geometries, and we do not need to impose a specific group structure.

The dynamical variables are now $(e^a_i,k^i_a)$, and the equations of motions are given by the first order differential system in $t$:

\begin{align}
 \label{eq:EOM_general_relativity_e}
 \partial_t e^a_i&=- \Omega_0^{ij} e^a_j - D_b \mathsf{N}^ae^b_i - \mathsf{N} e_i^b e_j^a k^j_b\\
 \partial_t k^i_a& = -{{\epsilon^i}_j}^k k^j_a \Omega_0^k + \mathsf{N}^b D_b k^i_a - \mathsf{N} \left({^3}R^i_a + k k^i_a - k^i_b k^b_a + \Lambda e^i_a + \kappa \left(T^i_a -\frac{T}{2}e^i_a\right) \right)
 \label{eq:EOM_general_relativity_k}
\end{align}
with $\Omega_0^{ij}= {{\epsilon^i}_j}^k k^j_a \Omega_0^k$ the coefficients $0$ of the spin-connection, which can be regarded as another gauge variable, and $T^i_a = T_{ab}e^{bi}$ is the stress-energy tensor restricted on $t=cst$ hypersurface.

The equation of motions gives us how a solution evolves, but we need additional constraints to specify if the couple $(e^a_i,k^i_a)$ is physically admissible. The equations of motion must conserve these constraints so that we can impose them only at the boundary. In the standard ADM formulation of general relativity, there are two constraints: the scalar and the vector constraint. The scalar constraint is used to derive the equations of motions \eqref{eq:EOM_general_relativity_e} and \eqref{eq:EOM_general_relativity_k}, and the second one imposes a structure on space-like hypersurface~\cite{rovelli_covariant_2014,misner_gravitation_1973,arnowitt_dynamics_1962}. We do not report them here since they do not play a role in our discussion.

We are now in a position to give the optimal control action for general relativity written for the new variables:
\begin{equation}
S_{OC} = h_{\partial V}[f,k] + \int_V d^4x ~\left( P^i_a ( \partial_t e^a_i [f] - \Upsilon^a_i[f,k] ) + \Pi^a_i(\partial_t k^i_a - \Gamma^i_a[f,k]) \right)
\label{eq:SOC_gravity}
\end{equation}
with $h_{\partial V}$ a boundary cost, and $\Upsilon, \Gamma$ are given by the rightmost term of Eq.~\eqref{eq:EOM_general_relativity_e} and \eqref{eq:EOM_general_relativity_k}. Note that under a diffeomorphism transformation, the adjoint states $P$ and $\Pi$ are not changed like a tetrad field, but they absorb a term $\text{det}(e^I_\mu)$ in order to maintain $S_{OC}$ gauge invariant. The physical interpretation of the O.C. adjoin state of the gravitational field is discussed in Sec.~\ref{sec:physical interpretation}.

\begin{proposition} \textbf{(First order variation of the gravitational OC action with respect to the gravitational adjoint states)}
\[
\frac{\delta S_{OC}}{\delta P^i_a(x)} = \Id_V(x) \left( \partial_t e^a_i [f] - \Upsilon^a_i[f,k] \right) + O(\epsilon^{q+1}).
\]
\[
\frac{\delta S_{OC}}{\delta \Pi^a_i(x)} = \Id_V(x) \left( \partial_t k^i_a - \Gamma^i_a[f,k] \right) + O(\epsilon^{q+1}).
\]
\label{eq:first_order_variation_SOC_with_grav_adj_state}
\end{proposition}

The variation of $S_{OC}$ with respect to the other variables is far less trivial, specifically with respect to $f$. The complexity in the calculations is of two orders: we have many second-order partial derivatives (since the curvature contains second-order variations of the metric tensor), and the calculation of variations of $e$ gives complicated formula since in general, we have:

\begin{proposition} \textbf{(First order variation of the gravitational OC action with respect to $f_A$ and $k^j_b$)} Without boundary cost, the first order variation of $S$ with respect to $f_A$ and $k^j_b$ is 
\begin{equation}
\begin{split}
\frac{\delta S_{OC}}{\delta f_A(x)} =&  - \left( -P^i_a (x) \frac{\delta \Upsilon_i^a}{\delta f_A(x)} + \Pi_i^a (x) \frac{\delta \Gamma^i_a}{\delta f_A(x)} + \partial_t P^i_a (x) C^{Aa}_i(x)\right) \Id_V(x) \\
&-  P^i_a (x) C^{Aa}_i(x) \frac{\partial \Id_V}{\partial t }(x) \\
\frac{\delta S_{OC}}{\delta k^j_b(x)} =&  - \left( -P^i_a (x) \frac{\delta \Upsilon_i^a}{\delta k^j_b(x)} + \Pi_i^a (x) \frac{\delta \Gamma^i_a}{\delta k^j_b(x)} + \partial_t \Pi^b_j (x) C^{Aa}_i(x)\right) \Id_V(x) \\
&-  \Pi^b_j (x)  \frac{\partial \Id_V}{\partial t }(x) 
\end{split}
\label{eq:first_order_variation_SOC_with_tetrad_2}
\end{equation}
\label{prop:first_order_variation_SOC_with_tetrad}
\end{proposition}

With this proposition, we see that we have strong constraints on the adjoint state and the boundary cost if we want to obtain $\delta S = 0$ with respect to all variables of the theory. Similarly to the harmonic oscillator, the value of the adjoint states is imposed at the boundary with variations of the boundary cost. However, the equation of motion inside $V$ is not equivalent to Eq.~\eqref{eq:first_order_variation_SOC_with_grav_adj_state}. This is because the equation of motions for $P$ and $\Pi$ are linearized versions of the equation of motion of $e$ and $k$. Without the explicit calculation of the entire system of equations, we can underline the effect of the linearization with a single non-linear therm of Eq.~\eqref{eq:EOM_general_relativity_k}. Let us examine what happens with $\mathsf{N}  k k^i_a$. The computation of $\tfrac{\delta \mathsf{N}  k k^i_a}{\delta k^j_b(x)}$ gives us $\tfrac{\delta \mathsf{N}  k k^i_a}{\delta k^j_b(x)} = \mathsf{N}(k^i_a e^b_j + k \delta^i_j \delta^b_a)$. Next, let us assume that in a part of the boundary, the terminal cost leads to the condition $\Pi^a_i \propto e^a_i$ or $\Pi^a_i \propto k^a_i$, then term $\Pi^a_i \tfrac{\delta \mathsf{N}  k k^i_a}{\delta k^j_b(x)}$ entering in the computation of Eq.~\eqref{eq:first_order_variation_SOC_with_tetrad_2} gives us $ \propto 2k e^b_j$ or $\propto k k^a_i + k^j_a k^a_i e^b_j$. None of these two terms appears in the equation of motion of $e$ and $k$ in \eqref{eq:EOM_general_relativity_e} and \eqref{eq:EOM_general_relativity_k}. As a consequence, $\Pi(t+dt)^a_i$ cannot be, in general, proportional to $k^a_i$ or $e^a_i$ in the entire volume $V$. 

A similar result can be derived if we consider the action of the form $\mathfrak{\int}_V d^4x P^{\mu\nu}(G_{\mu\nu}(x) - \kappa T_{\mu\nu}(x))$, where $P^{ \mu \nu}$ is the adjoint state to EFE, a 4 by 4 matrix. The equation of motion of $P_{\mu\nu}$ cannot be reduced to EFE due to the linearization, and thus, we cannot impose, in general, $P^{\mu\nu} \propto g^{\mu\nu}$ for all points in $V$.

\section{Optimal Control Path Integrals With Generalized Functions}
\label{sec:path_integral_oc}

In the previous sections, we studied the calculus of variation of physical systems with generalized functions and optimal-control type actions. We have underlined the nice regularity properties of optimal control actions and we have emphasized that a quantum scalar product allows us to fix the value of the adjoint state on the boundary of the integration volume. In this section, we go one step further to the quantum theory and we develop a path integral for the optimal control action.

\subsection{General Considerations}

In this section, we define a path integral using generalized functions. More precisely, we give a well behave description of
\[
\int \Mc D[\phi] e^{ \ii S [\phi]}.
\]

The core of the approach is the same as in the usual approach of the path integral~\cite{zinn2004path}, the fields must be discretized into a finite number of degrees of freedom, and next, the limit when this number of degrees of freedom becomes infinite can be taken. The main difference is that we work with generalized functions which are based on smooth functions. Thus, we keep the regularity of the fields, as required in the calculation of most of the actions.

The discretization of the field can be performed in different ways. The simplest one is given by the theorem of approximation \ref{thm:approx}, but we have endless ways to discretize a continuous function. We can use standard tools of Hilbert spaces, with a decomposition of function into a given basis, use wavelets~\cite{shi1999generalized} or splines~\cite{ferguson1964multivariable}. Once a method is chosen, we can map the function into $\Mc E_M$ using the convolution product with a mollifier. The choice of a discretization procedure as well as the choice of the mollifier may depend on the problem considered.

In the following, focus on a system described with an OC action~\ref{eq:action_oc}. The standard case when the Lagrangian is quadratic is considered in appendix~\ref{sec:appendix_path_integral_ho}.

The OC action is defined on the extended space of configurations, and thus, in the path integral, we have to integrate over $p$ and $q$, which are assumed to be quantum fields. The control $u$ is treated as a classical field, but we can easily promote it to a quantum field. The time axis is discretized into $N-1$ intervals $[t_n,t_{n+1}]$ such that $t_1 = t_i$ and $t_f = t_N$. For simplicity, we choose $\Delta t = t_{n+1} - t_n$ constant, but this is not strictly necessary.

For simplicity, we choose the following discretization scheme for the fields::
\begin{align}
\label{eq:discretized_p_oc}
p_\mu (t)& = \sum_{n=1}^{N-1} p_{n,\mu} \Delta t ~ \eta_\epsilon(t-t_n), \\
\label{eq:discretized_q_oc}
q^\mu(t) & = \int d\tau ~ \eta_{\epsilon} (t - \tau)\sum_{n=1}^{N} q_n^\mu \underbrace{\left( \Id_n (\tau) \left[ 1 - \frac{\tau-t_n}{\Delta t}\right]+ \Id_{n-1} (\tau) \frac{ \tau - t_{n-1}}{\Delta t} \right)}_{\Delta_n (\tau)} \\
& = \int d\tau ~ \eta_{\epsilon} (t - \tau)\sum_{n=1}^{N} q_n^\mu ~\Delta_n(\tau) \\
& = \sum_{n=1}^{N} q_n^\mu ~\tilde \Delta_n(t),\\
\label{eq:discretized_u_oc}
u^i (t)& = \sum_{n=1}^N u_{n}^i \Delta t ~ \eta_\epsilon(t-t_n) ,
\end{align}
where the notation $\Id_n(t) \equiv \Id_{[t_{n},t_{n+1}[}(t)$ is used. Note that $p_\mu$ and $u^i$ are discretized with piecewise constant functions while $q^\mu$ is discretized with a linear spline. This choice is made for simplicity of presentation, but $p_\mu$ and $u^i$ can be also discretized with a linear spline, at the price of unnecessary subtleties for the following presentation.

We define the path integral as follows:
\begin{definition}\textbf{Optimal control path integral}
Let $S$ be the action of the optimal control system Eq.~\eqref{eq:action_oc}, defined on the configuration space $\Mc C = \Gamma \times \Gamma \times U = \{ (q(t),p(t),u(t))~|~q(t) \in \Gamma, p(t) \in \Gamma, u(t) \in U \}$. Let the discretization scheme given by Eq.~\eqref{eq:discretized_p_oc} to Eq.~\eqref{eq:discretized_u_oc}. Let $\psi(q_1,p_1)$ be a function of the system state and its adjoint state at the initial time. The propagator of $\psi$ as a function of time is given by:
\[
\hat W(t_1,t_N)\psi(q_1,p_1) = \frac{1}{(2 \pi)^{n(N-1)}}\int_{\Gamma^{2(N-1)}} \prod_{n,\mu,\nu} dq_n^\mu d p_{n,\nu} ~ e^{i S(\{q_n^\mu,p_{n,\nu},u_n^i\})} \psi(q_1^\mu,p_{1,\nu}). 
\]
\label{def:oc_path_integral}
\end{definition}
In order to show that $\hat W$ is indeed a propagator of the system, we shall perform a few calculations. First, we have to derive an explicit formula for $S(\{q_n^\mu,p_{n,\nu},u_n^i\})$. Using the discretization schemes Eq.~\eqref{eq:discretized_p_oc} to Eq.~\eqref{eq:discretized_u_oc}, and inserting into Eq.~\eqref{eq:action_oc}, we obtain:
\begin{equation}
\begin{split}
S =& h_f \left( \sum_{n=1}^N q_n^\mu ~\tilde \Delta_n(t_N) \right) - h_i \left( \sum_{n=1}^N q_n^\mu ~\tilde \Delta_n(t_1) \right) \\
 & + \sum_{n=1}^{N-1}  p_{n,\mu} \left( q_{n+1}^\mu - q_n  - \Delta t~ f^\mu \left( q_n \right) \right) + O(\epsilon^{q+1})\\
 & + \underbrace{\int dt f^0 (q(t),u(t),t)}_{S^0}
\end{split}
\label{eq:new_form_action_oc}
\end{equation}
The last term is left unevaluated since there is no simple expression if the explicit form of $f^0$ is not fixed.

 Next, we can insert \eqref{eq:new_form_action_oc} into the definition of the path integral, and we note that the integration over $p_{n,\mu}$ gives us an inverse Fourier transform which gives a Dirac distribution, except with $p_{1,\nu}$, due to the presence of $\psi(q_1^\mu,p_{1,\nu})$ in the integral. Introducing $\Psi$ the Fourier transform of $\psi$ with respect to its second argument, we obtain:
\begin{equation}
\begin{split}
\hat W(t_1,t_N)\psi(q_1,p_1) =& \int_{\Gamma^{N-1}} \prod_{\mu} dq_1^\mu \prod_{n=2}^{N-1} dq_n^\mu ~\delta\left(  q_{n+1}^\mu - q_n^\mu  - \Delta t~ f^\mu \left(q_n \right) + O(\epsilon^{q+1})\right) \\
& \times e^{i (h_f - h_i + S^0)} \Psi\left(q_1^\mu,q_{2}^\mu - q_1^\mu  - \Delta t~ f^\mu \left(q_1  \right) + O(\epsilon^{q+1})\right). 
\end{split}
\label{eq:new_form_path_integral_oc}
\end{equation}
To simplify this equation, we can introduce $W $ the flow associated with the discretized system, such that $q^\mu_{n+1}= W^{\mu} (q_n)= q_n^\mu + \Delta t f^\mu (q_n)$, the Eq.~\eqref{eq:new_form_path_integral_oc} becomes:
\begin{equation}
\begin{split}
\hat W(t_1,t_N)\psi(q_1,p_1) =& \int_{\Gamma^{N-1}} \prod_{\mu} dq_1^\mu \prod_{n=2}^{N-1} dq_n^\mu ~\delta\left(  q_{n+1}^\mu - W^\mu(q_n) + O(\epsilon^{q+1})\right) \\
& \times e^{i (h_f + h_i + S^0)} \Psi\left(q_1^\mu,q_{2}^\mu - W^\mu(q_1^\mu)+ O(\epsilon^{q+1})\right). 
\end{split}
\label{eq:new_form_path_integral_oc_2}
\end{equation}
Next, we can integrate the Dirac distributions. For that purpose, consider the general case:
\begin{equation}
\int dx ~ \delta(F(x) + \mathfrak{n}(x)) ~ \phi(x)
\label{eq:integral_dirac_distribution}
\end{equation}
where $F$ is moderate and $\mathfrak{n}$ is negligible, of order $\epsilon^{q+1}$, $\phi$ is an arbitrary generalized function. We introduce $G(x) = F(x) + \mathfrak{n}(x)$, and we assume that $G$ is analytic and has a countable number of zero. Let $I=\{x_0 ~|~G(x_0) = 0\}$. From Eq.~\eqref{eq:integral_dirac_distribution} we deduce:
\begin{equation}
\int dx~ \delta(G(x)) \phi(x) = \sum_{x_0 \in I} \int_{\Mc V (x_0)} dy~ \delta \left( G(x_0) + \sum_{n=1}^\infty \frac{1}{n!} \left.\frac{\partial^n G}{\partial x^n}\right\vert_{x=x_0} y^n \right) \phi(x_0+y)
\end{equation}
where the integration is over neighborhoods $\Mc V$ of $x_0$, and the restriction of $G$ on $\Mc V (x_0)$ is a bijective mapping. After a last change of variable, we obtain:
\begin{equation}
\int dx~ \delta(G(x)) \phi(x) = \sum_{x_0 \in I} \left( \left.\frac{\partial G}{\partial x}\right\vert_{x=x_0}\right)^{-1} \phi(x_0).
\end{equation}
Here, the Dirac distribution must be taken exactly, it is not regularized with a mollifier, hence $x_0$ depends on both $F$ and $\mathfrak{n}$. With this result, we see that the discretization scheme can have an impact on the path integral. We must have $\tfrac{\partial G}{\partial x}\vert_{x=x_0}\neq 0$, to avoid divergence, and we must have the existence of at least one $x_0$ to obtain a continuous limit, otherwise, we may have $\delta(G(x)) = 0$ even for $\epsilon \rightarrow 0$. Then, we must have $F(x_0)= o (\epsilon^{q+1})$. Now, let us assume that $\lim_{\epsilon \rightarrow 0} x_0(\epsilon) = X_0$, and as a consequence, $\lim_{\epsilon \rightarrow 0} G(x_0(\epsilon)) = F(X_0) = 0$. A straightforward calculation shows us that $|x_0(\epsilon) - X_0|$ is of order $\epsilon^{q+1}$, and hence,
\begin{equation}
\int dx~ \delta(G(x)) \phi(x) = \sum_{x_0 \in I} \left( \left.\frac{\partial F}{\partial x}\right\vert_{x=X_0}\right)^{-1}  \phi(X_0) + O(\epsilon^{q+1}).
\end{equation}
The discretization scheme Eq.~\eqref{eq:discretized_q_oc} have these suitable properties, which is not necessarily the case of a discretization scheme similar to Eq.~\eqref{eq:discretized_p_oc} where $x_0$ may not exist.

Returning to our initial problem, we can integrate the Dirac distributions of Eq.~\eqref{eq:new_form_path_integral_oc_2}. Assuming that there exists a unique $q_{n+1}$ such that $q_{n+1} = W(q_n)~\forall n$, we have $\tfrac {\partial}{\partial q_{n+1}^\mu}  (q_{n+1}^\mu - q_n - \Delta t f^\mu(q_n)) = 1$ and thus  $\int dq_{n+1}^\mu \delta(q_{n+1}^\mu - W^\mu(q_n) + O(\epsilon^{q+1})) \psi(q_{n+1}^\mu, q_n) = \psi(W^\mu(q_n),q_n) + O(\epsilon^{q+1}) $. There is no general simplification of Eq.~\eqref{eq:new_form_path_integral_oc_2}, but a simpler formula can be derived by computing:
\begin{equation}
\braket{\psi'(q_N,p_N)}{\hat W(t_1,T_N) \psi(q_1,p_1)} = \int_{\Gamma^2}  \prod_{\mu \nu} dq_N^\mu dp_{N,\nu} \psi'^\dagger(q_N,p_N) \hat W(t_1,T_N) \psi(q_1,p_1).
\end{equation} 
A straightforward calculation using the flow $W$ gives us:
\begin{equation}
\begin{split}
\braket{\psi'(q_N,p_N)}{\hat W(t_1,T_N) \psi(q_1,p_1)} = &\int_{\Gamma^3} \prod_{\mu,\nu} dq_1^\mu dq_2^\mu dp_{N,\nu} ~\Psi\left(q_1^\mu,q_{2}^\mu - W_1^\mu \left(q_1  \right) + O(\epsilon^{q+1})\right) \\
& \times e^{\ii~(h_f(q_1,q_2) + h_i(q_1) + S^0(q_1,q_2)} ~\psi'^\dagger(W^\mu_{N-2}(q_2),p_{N,\nu})+ O(\epsilon^{q+1})
\end{split}
\end{equation}
where $W_{N-2}^\mu$ is the flow ittreatred $N-2$ times. 
The specific case when $\psi$ and $\psi'$ are not functions of $p$ is particularly interesting since the Fourier transform with $p_{1,\nu}$ gives us a Dirac distribution, and we have $\Psi(q_1^\mu) = \psi(q_1^\mu) \delta (q_2^\mu - W^\mu(q_1) + O(\epsilon^{q+1})$. The final result is therefore simple:
\begin{equation}
\braket{\psi'(q_N)}{\hat W(t_1,T_N) \psi(q_1)} = \int_\Gamma \prod_\mu dq_1^\mu \left[ \psi(q_1^\mu) ~\psi'^\dagger(\phi^\mu_{N-1}(q_1)) ~e^{\ii~(h_f(q_1)+h_i(q_1) + S^0(q_1)}\right]  + O(\epsilon^{q+1})
\end{equation}
This last equation can be interpreted as the scalar product between two states, where the variables have been propagated using the classical equation of motion. This is a property expected with coherent states~\cite{zinn2004path,gardiner2004quantum}. The exponential $e^{\ii~(h_f(q_1)+h_i(q_1) + S^0(q_1)}$ plays the role of a density in the scalar product. 

The role taken by the adjoin state in the quantum theory is discussed in Sec.~\ref{sec:physical interpretation}.

\subsection{Generalization to a Field Over Space-time}
\label{sec:Generalization with a field over space-time}

We investigate in this paragraph how the path integral is modified when the action \eqref{eq:action_oc} is defined by a tensor field $T^\mu$ over space-time. We focus on the term with the adjoint state, which can be written as:
\begin{equation}
S = \int_{\Mc M} d^4x~ p_{\mu}^\nu (x)\left[ \partial_\nu T^\mu(x) - f_\nu ^\mu(T(x))\right]
\label{eq:action_OC_fiel_over_space-time}
\end{equation}
Note that the covariant derivative is not explicitly introduced here, but its effect can be hidden into $f_\nu^\mu$. From this equation, it is clear that the configuration space is $\Gamma \times \Gamma ^2$, and in the path integral, we have to integrate over $T^\mu$ and $p_\mu^\nu$. Using a similar discretization scheme as before, we find easily that at a given space-time location $x$, the contribution to the path integral is of the form $\prod_{\mu \nu} \delta \left(T^\mu(x + \Delta x_\nu) - T^\mu(x ) + \Delta x_\nu f_\nu^\mu (T(x)) \right)$. We keep exactly the same properties with a field over $\setR$, the only difference being that the boundary state $\psi$ is defined on the boundary of a closed space-time region which is suitable for the integration of the equation of dynamics (i.e. we need a well-posed Cauchy problem).

\subsection{Action Under a Diffeomorphism}

We consider an optimal control action over space-time in the form of Eq.~\eqref{eq:action_OC_fiel_over_space-time}. We recall that $x$ is a coordinate system of $\Mc M$, and $p$, $q$, $f$ are generalized fields. Then, under the action of a diffeomorphism, they must transform using the diffeomorphism operator introduced in Definition~\ref{def:diffeomorphism_operator}. However, there are two small subtilities since we impose that $S$ is invariant by a change of coordinate system. This leads to the fact that $p$ is not a generalized tensor field, but a generalized tensor density and $f$ is a tensor field that is a function of another tensor field. Hence, the transformation rules are the following:
\begin{equation}
p^{\nu'}_{\mu'} (X)  = \frac{1}{|\det (\partial f)|(X)} \left[\hat D (x,X) p\right]^{\nu'}_{\mu'} + O(\epsilon^{q+1}) 
\end{equation}
\begin{equation}
\begin{split}
f^{\mu'}_{\nu'} (T(X)) = & \left[\hat D (x,X) {f_{(0)}}\right]^{\mu'}_{\nu'} +  \left[\hat D (x,X) {f_{(1)}}\right]^{\mu'}_{\nu' \sigma'}  \left[\hat D (x,X) T\right]^{\sigma'}   \\
& +  \left[\hat D (x,X) {f_{(1)}}\right]^{\mu'}_{\nu' \sigma' \rho'}   \left[\hat D (x,X) T\right]^{\sigma'} \left[ \hat D (x,X) T\right]^{\rho'}  + ... + O(\epsilon^{q+1}).
\end{split}
\end{equation}
Using these transformations, we have:
\begin{equation*}
\int_{\Mc M} d^4x~ p_{\mu}^\nu (x) \left[ \partial_\nu T^\mu(x) - f_\nu ^\mu(T(x))\right] = \int_{\Mc M} d^4X~ p_{\mu'}^{\nu'} (X) \left[ \partial_\nu T^{\mu'}(X) - f_{\nu'} ^{\mu'}(T(X))\right] + O(\epsilon^{q+1}).
\end{equation*}

\subsection{Application to Physical Systems}

In this section, we explore briefly, the application of the optimal control path integral for the physical systems studied in Sec.~\ref{sec:anlysis_action_physical_syst}. In all cases, we use $S_{OC}$ as a starting point, and we consider the case where the quantum state does not depend on the adjoint state.

Since all these actions are of the form $S_{OC} = h(q) + \int dt p(\dot{q} - f(q))$, the implementation of the theory is quite straightforward. 

\subsubsection{Optimal Control path integral for the Harmonic Oscillator}
\label{sec:opt_path_integral_ho}

We consider the action \eqref{eq:action_oc_ho}, with a terminal cost only at the final time (this is only a matter of simplicity in the equations). The field $\alpha$ is discretized with a linear spline, and the adjoint state with a piecewise constant function. The path integral is 

\begin{equation}
\begin{split}
\hat W(t_1,t_N) \psi(\alpha_1) &= \frac{1}{(2 \pi)^{2(N-1)}}\int_{\setC^{4(N-1)}} \prod_{n} d\alpha_n d\bar \alpha_n d \pi_n d\bar \pi_n ~ e^{i S_{OC}(\{\alpha_n,^ip_{n},u_n^i\})} \psi(\alpha_1)\\
& = \int d\alpha_1 d\bar \alpha_1~ \psi(\alpha_1) \braket{\beta_f}{W^{t_N-t_1}(\alpha_1)} .
\end{split}
\end{equation}
Where $\beta_f$ comes from the definition of the boundary cost function.
Using $\psi(\alpha_1) = \delta(\alpha_1 - \beta_i)$ gives us:
\begin{equation}
\hat W(t_1,t_N) \psi(\alpha_1) = \braket{\beta_f}{W^{t_N-t_1}(\beta_i)} ,
\end{equation}
which is the scalar product between two coherent states.

\subsubsection{Optimal Control Path Integral for the Scalar Field}

Similarly to the harmonic oscillator, we provide a path integral for the scalar field based on the optimal control action. The extended configuration space being $\Mc C = \setR^{2(4+1)}$ (one variable for the scalar field, four for its derivatives, and the same numbers for the adjoint state), we have to integrate over this space at each space-time position. The boundary term in the action is given by Eq.~\eqref{eq:bdr_scalar_field}. Contrary to the harmonic oscillator, we have space and time coordinates, thus the $N$ discretization points are not only on the time axis but on the entire volume of definition. The boundary of $V$  is cut into two parts, noted $\partial V_1$ and $\partial V_2$, for \textit{in} and \textit{out} states. In these boundaries, there are respectively $M$ and $M'$ points. The discretization scheme is similar to the harmonic oscillator, piecewise constant for the adjoint state, and linear interpolation for the field variables. Explicitly, we have:
\begin{equation}
\begin{split}
\hat W_{\partial V_1 \mapsto \partial V_2} \Psi(\{\psi_m,P_{\mu,m}\}_{m \in \partial V_1}) &= \frac{1}{(2 \pi)^{5(N-M')}}\int_{\setR^{10(N-M')}} \prod_{n,\mu,\nu} d\psi_n d\phi_n d P_{\mu,n} d \pi_n^\nu  \\
& ~ ~ ~ \times e^{i S_{OC}(\{\psi_n,\phi_n,P_{\mu,n},\pi_n^{\nu}\})} \Psi(\{\psi_m,P_{\mu,m}\}_{m \in \partial V_1})
\end{split}
\end{equation}
Note that $\partial V_1$, $\partial V_2$, and their associated points cannot be chosen totally arbitrarily. We need that the data of the field at the $M$ points allow us to reconstruct the field at the $M'$ other points. Briefly, this fact can be seen with a 2-dimensional reduction of the manifold (for example, keeps only the time coordinate and one space coordinate). With a first-order finite difference approximation of the derivatives, we easily see that the field at a point in the slice of coordinate $t+\Delta t$ can be determined with the field configuration of two adjacent points in the slice of coordinate $t$. From the point of view of the continuum limit, this is similar to the fact that the field at a given space-time position is determined by a continuum of points at a previous time, located in the past light-cone. To simplify the equation, we introduce the shorthand notation $\Psi_{\partial V_i}$ to specify that $\Psi$ depends on points in $\partial V_i$, and we do not write the variables of the function. Then, the path-integral reads:

\begin{equation}
\hat W_{\partial V_1 \mapsto \partial V_2} \Psi_{\partial V_1} = \int_{\setR^{5M}} \prod_{n,\mu,\nu} d\psi_n d P_{\mu,n}  \braket{\zeta_{\partial V_2}}{W_{\partial V_1 \mapsto \partial V_2}(\{\psi,P_\mu\})} \Psi_{ \partial V_1}
\end{equation}
with $\zeta$ coming from the definition of the boundary cost function (see Eq.~\eqref{eq:bdr_scalar_field}). Exactly like the harmonic oscillator (see Sec.~\ref{sec:opt_path_integral_ho}), if $\Psi$ is a Dirac distribution in the space of field variable, it corresponds to a coherent state, and we recover the scalar product between coherent states.

\subsubsection{Optimal Control Path Integral for Lorentzian Space-time}

The construction of the path integral for Lorentzian space-time follows the same lines as the one of the scalar field, with the same subtleties concerning the boundary of the integration volume. Here, we provide a few comments on the differences. One of the main differences is that we can perform a 3+1D splitting, and we are reduced to a first-order differential equation on the time axis, but we still have a second-order differential equation (in the curvature tensor of a 3-surface). Hence, the discretization scheme must be, at least, given by a first-order spline on the time axis, and a second-order spline in space. As discussed in Sec.~\ref{sec:lorentzian_space_time}, the field variables that must be used for the calculus of variation are $f_A$ and $k^j_b$, and thus, we have to integrate over these fields at each position of the discretization grid (as well as their corresponding adjoint states).

Nevertheless, there is a subtle point concerning the choice of integration measure. Due to the nature of $f_A$ and $k^j_b$, we may choose a Lebesgue measure for all the variables of integration. However, we may also choose an integration measure that coincides with the Haar measure on $GL_3(\setR)_+$, like with the loop quantum gravity formalism which uses abundantly 
 integrals over $SU(2)$ or $SL(2,\setC)$ \cite{rovelli_quantum_2007,rovelli_covariant_2014}. The choice of integration measure has necessarily an impact on the integral of the function $\psi$ at the boundary, and thus, there is a potential influence on the boundary term used in the action, and as a consequence, an influence on the definition of the scalar product between coherent states. The clarification of this point is left for another study. Note that the problem has been tackled in~\cite{ansel2022loop}, where the gravitational system is written in terms of real Ashtekar variables, and the boundary cost function is defined using complexifier coherent states, which are standard coherent states used in loop quantum gravity and in spinfoam theory.

\section{Physical interpretation of the OC adjoint state}
\label{sec:physical interpretation}

The introduction of an optimal control adjoint state into the model of the physical system is interesting from a mathematical point of view, and it regularizes greatly the action variations at the boundary. Additionally, it generalizes the standard actions where these latter can be recovered with a specific value of the adjoint state. Despite these interesting mathematical properties, it must be clarified the role taken by adjoint state can play in a physical system, and how it can be interpreted.

From the definition of the optimal control Lagrangian \eqref{eq:action_oc}, it is clear that this is a quantity that allows us to measure how much $\dot q^\mu$ can deviate from $f^\mu(q,u,t)$. If the action is zero for an arbitrary value of $p^\mu$ (ignoring the contributions of $f_0$ and the boundary costs), the system follows a classical trajectory, and otherwise, it doesn't. The presence of the adjoint state is therefore intimately connected to the notion of an observer. It is a quantity that allows the observer to identify the trajectory of the system. With this formulation of the action principle, the system is not defined "by - itself", but makes sense with respect to another system, described by $p$. A similar interpretation can be made in the case of general relativity. In Eq.~\eqref{eq:SOC_gravity}, the adjoint states allow us to "measure" Einstein's field equations. From the standard point of view, EFE makes sense because we have an underlying structure of manifold. In order to define this equation, first, we must provide the tangent space of the manifold at a given position. In general, we describe the tangent space with the directional derivatives of the coordinates, but this is not strictly necessary. The gravitational adjoint state is a way to incorporate the notion of tangent space into the action. This is an interesting feature since we did not suppose any manifold structure for the gravitational field, we have just imposed that the triad at a given position must be in $GL_3(\setR)$. The structure of a manifold is rather a consequence of the least action principle.

When the path integral is taken, we integrate over all the possible "observer" state $p$, and the dynamics are constrained to follow classical trajectories. The quantum phenomena are therefore pushed at the boundary of the domain of definition. The full quantum state at the boundary is given in general by a function $\psi(q,p)$ which can describe how the quantum system and the observer have joined information. In a situation when they are not entangled, we can trace out the observer degrees of freedom and we are left with a quantum state $\psi(q)$ which is the quantum state of the system alone.

\section{Conclusion}
\label{sec:conclusion}

In this paper, we have introduced notions of generalized functions to the calculus of variation and path integrals. A specific analysis has been carried out for optimal control actions, whose Lagrangians are linear functionals of the equation of motion, instead of quadratic functions of the dynamical quantities.

In the first step, we introduced basic notions of generalized function, and then a proposal of generalized tensor fields has been given. This proposal differs from the one developed recently~\cite{nigsch2020nonlinear} because the Colombeau algebra remains coordinate dependent. Diffeomorphism invariance is recovered by the use of an operator that maps together Colombeau algebras defined with different coordinate systems. With this proposal, the application of the calculus of variation, usually defined with Gateaux derivatives, is straightforward. This construction is particularly useful since the calculus of variation is coordinate-dependent in the sense that it returns a tensor field written in a given coordinate system, but it is also diffeomorphic invariant in the sense that the procedure can be performed equivalently in any coordinate system.

Next, the calculus of variation has been revisited with the new formalism, and first- and second-order variations of a functional have been investigated. In particular, the analysis of the functional variation on the boundary of the integration volume has been performed with caution. This kind of calculation is ill-defined with distributions since it requires the computation of products of distributions, but all the calculation steps make sense with generalized functions. At the boundary, the variation of a functional is usually a divergent generalized function. As a side comment, the fact that the sign of the second-order variation depends on the mollifier was analyzed briefly. This feature leads us to additional assumptions on the mollifier to derive the Legendre-Clebsh condition for optimality.

Then, the paper focuses on the application of the formalism to three case studies: the harmonic oscillator, the scalar field, and the gravitational field (modeled by a Lorentzian manifold). For the first two cases, we compare three types of actions leading to the same equations of motion, and hence, the same classical physics. The first action is the standard action encountered in the literature, the second one is the optimal control action, and the third one is an action similar to the optimal control action, but with second-order derivatives instead of first-order ones. We prove that a specific choice of the boundary cost function in the optimal control action allows us to impose a condition on the adjoint state so that the adjoint state can be identified as the state of the system. For the harmonic oscillator, the boundary cost that verifies this property is given by the scalar product between coherent states. This gives us a natural connection between classical and quantum theory, where classical physics is observed inside the volume of integration, and quantum theory is observed at the boundary of the volume. For the scalar field, the case of the harmonic oscillator can be adapted, and the boundary cost function can be identified as the usual scalar product of coherent states when the boundary is the union of space-like hypersurfaces. It shall be stressed that the quantization of the scalar field is subject to many issues, specifically in curved space-time (such as the construction of Hadamard states). This kind of issue has been totally disregarded here, and additional work is required to fully justify that the proposal of boundary cost is always a well-defined scalar product of quantum states. In the case of a Lorentzian manifold, the calculus of variation requires more caution, to preserve the properties of the metric tensor. For that purpose, methods taken from the Lie-group theory have been used. The calculations of the first-order variations of the optimal control action for a Lorentzian manifold have led us to the conclusion that the boundary cost function cannot impose the adjoint state to take the value of the gravitational field variables. This is because the equations of motion of the adjoint state are a linearization of Einstein’s field equations. Finally, for the third type of action, we have underlined that a similar construction as the optimal control action can be performed, but the boundary cost function must be modified and it is not clear yet if it can be assimilated to a quantum scalar product.

The paper finishes with the construction of a path integral with the optimal control action. Interestingly, generalized functions allow us to handle simultaneously discretized and smooth functions. Then, the path integral is defined, as usual, with a discretization grid, but the field modeling of the physical system remains a continuous function. A direct calculation of the integrals at each point of the discretization allows us to show that the path integral gives us a propagator for coherent states, as already found in \cite{ansel2022loop}. The propagator is entirely determined by the flow of the classical equation of motion. In the case when the quantum state does not depend on the optimal control adjoint state, the path integral at the boundary of the integration volume is reduced to the scalar product of coherent states. If the state depends on the adjoint state, we have a non-trivial modification of the quantum state, which will be studied in detail elsewhere. The use of a generalized function in the calculation of the path integral allows us to analyze precisely the convergence of the integral in the continuous limit of the discretization grid. This latter must be chosen with some caution; otherwise, the path-integral is identically zero.

\section*{Acknowledgements}
The author acknowledges D. Viennot for useful feedback on this work


\begin{appendix}

\section{Example of second order variation: the Legendre-Clebsh condition revisited}
\label{sec:Legendre-Clebsh condition}

In this appendix, we revisit the Legendre-Clebsh condition~\cite{Courant_Hilber_vol_1,Courant_Hilber_vol_2}, as an example of a second-order variation of functional with non-linear generalized functions. A specific look is addressed on the additional assumptions that we have to make on the field variations.

\begin{theorem}
\label{thm:2nd_oder_variation_and_optimality_condition}
Let $I$ be a functional defined by $I= \int
 d^Nx F(\tilde T^\mu(x), \partial_\nu \tilde T^\mu(x))$. For simplicity, we assume that the functional is over an unbounded set.
Under the assumption that $d_\mu \eta (0) = 0$ and $d_\mu d_\nu \eta (0) >0$,  
the second-order variation is a (divergent) generalized number given by:
\begin{equation*}
\begin{split}
\frac{\delta^2 I}{\delta \tilde T^\mu\delta \tilde T^\nu}(y) = &\left[ \frac{\partial^2 F}{\partial \tilde T^\mu \partial \tilde T^\nu}(y) - \frac{d}{dy^\sigma}\left( \frac{\partial^2 F}{\partial \tilde T^\mu \partial(\partial_\sigma \tilde T^\nu)} (y)\right)\right] \frac{\eta (0)}{\epsilon}  \\ 
&- \frac{\partial^2 F}{\partial(\partial_\rho \tilde T^\mu)\partial(\partial_\sigma \tilde T^\nu)}(y) \frac{d_\sigma d_\rho \eta (0)}{\epsilon^2} + O(\epsilon^{q+1}),
\end{split}
\end{equation*}
and for $\epsilon$ sufficiently small, the sign of $\frac{\delta^2 I}{\delta \tilde T^\mu\delta \tilde T^\nu}(y)$ is given by the sign of
\[
\frac{\partial^2 F}{\partial(\partial_\rho \tilde T^\mu)\partial(\partial_\sigma \tilde T^\nu)}(y).
\]
\end{theorem}
\begin{proof}
From the definition \ref{def:n-th_order_variation_functional}, we determine:
\begin{equation}
\begin{split}
\frac{\delta^2 I}{\delta \tilde T^\mu\delta \tilde T^\nu}(y) = \int d^Nx & \left[ \frac{\partial^2 F}{\partial \tilde T^\mu \partial \tilde T^\nu}(x) \eta_\epsilon (x-y)^2 + \frac{\partial^2 F}{\partial \tilde T^\mu \partial(\partial_\sigma \tilde T^\nu)} (x) \eta_\epsilon(x-y) d_\sigma \eta_\epsilon(x-y)\right. \\
&+ \left.  \frac{\partial^2 F}{\partial(\partial_\rho \tilde T^\mu)\partial(\partial_\sigma \tilde T^\nu)}(x) \left(d_\sigma \eta_\epsilon (x-y) \right)\left(d_\rho \eta_\epsilon (x-y) \right) \right]
\end{split}
\end{equation}
The next step of the proof is to use integration by parts to transform $d_\sigma \eta_\epsilon$ into $\eta_\epsilon$. This allows us to determine:
\begin{equation}
\begin{split}
\frac{\delta^2 I}{\delta \tilde T^\mu\delta \tilde T^\nu}(y) = \int d^Nx & \left[ \frac{\partial^2 F}{\partial \tilde T^\mu \partial \tilde T^\nu}(x) - \frac{d}{dx^\sigma}\left( \frac{\partial^2 F}{\partial \tilde T^\mu \partial(\partial_\sigma \tilde T^\nu)} (x)\right)\right] \eta_\epsilon (x-y)^2 \\
&-\eta_\epsilon (x-y)\frac{d}{dx^\sigma} \left[ \frac{\partial^2 F}{\partial(\partial_\rho \tilde T^\mu)\partial(\partial_\sigma \tilde T^\nu)}(x)d_\rho \eta_\epsilon (x-y) \right]
\end{split}
\end{equation}
The final result is obtained by evaluating the integral up to the order $\epsilon^q$ using the moment properties of the mollifiers, after that, we use the assumption $d_\sigma \eta_\epsilon (0) = 0$ to simplify the formula. The last part of the proof is to notice that for $\epsilon$ sufficiently small, this is the term in $1/\epsilon ^2$ which dominates, and thus this is the sign of this term which gives the sign of $\delta^2 I$.
\end{proof}

\section{Quadratic path-integral}
\label{sec:appendix_path_integral_ho}

In this section, we revisit briefly, with the formalism of generalized functions, several aspects of the quadratic action of the harmonic oscillator and we compute the path integral. We recall that the action is:
\begin{equation}
S = \frac{1}{2}\int_{t_i}^{t_f} dt ~\left( m \left( \frac{dq}{dt}\right)^2 - k q^2 \right)
\label{eq:appendix_quadratic_action}
\end{equation}
\subsubsection{First and second order variations of the action}

The first  order variation of Eq.~\eqref{eq:appendix_quadratic_action} is:
\begin{equation}
\frac{\delta S}{\delta q(t)}=  -\left( m \frac{d^2 q}{dt^2} + k q(t) \right) \tilde{\Id}_{[t_i,t_f]}(t) + \frac{m}{2} \frac{dq}{dt} \frac{d \tilde \Id_{[t_i,t_f]}(t)}{dt} + O(\epsilon^{q+1})
\end{equation}
while the second order is:
\begin{equation}
\begin{split}
\frac{\delta^2 S}{\delta q(t_1)\delta q(t_2)} = &- m \left(  \left.\frac{d^2 \eta_\epsilon}{dt^2} \right\vert_{t= t_1-t_2} \Id_{[t_i, t_f]}(t_1) + \left.\frac{d \eta_\epsilon}{dt} \right\vert_{t= t_1-t_2} \left.\frac{d \Id_{[t_i,t_f]}}{dt} \right\vert_{t= t_1-t_2} \right) \\
&- k \eta_\epsilon (t_1-t_2) + O(\epsilon^{q+1})
\end{split}
\end{equation}
In most of the usual mollifiers, we have $\left.\tfrac{d \eta_\epsilon}{dt} \right\vert_{t=0} = 0$ and $ \left.\tfrac{d^2 \eta_\epsilon}{dt^2} \right\vert_{t= t_1-t2} <0$, hence for a usual variation of the field, we have
\begin{equation}
\sign\left[\frac{\delta^2 S}{\delta q(t)^2} \right] = 1.
\end{equation}
\subsubsection{Path integral}

For this path integral, we proceed similarly as in Sec.~\ref{sec:path_integral_oc}. We discretize $q$ with a finite number of degrees of freedom. In the following, it is enough to use an approximation of the form of Thm.~\ref{thm:approx}:
\begin{equation}
q(t) = \sum_{n=1}^N q_n \Delta t ~\eta_\epsilon(t-t_n),
\end{equation}
with $t_{n+1} - t_n = \Delta t$. Inserting this approximation scheme into the action \ref{eq:appendix_quadratic_action}, we have:
\begin{equation}
\begin{split}
S = -\frac{1}{2} \sum_{ n l } & q_n q_l \Delta t^2 \left(  m\left.\frac{d^2 \eta_\epsilon}{dt^2} \right\vert_{t= t_n-t_l} \Id_{[t_i, t_f]}(t_n) + m \left.\frac{d \eta_\epsilon}{dt} \right\vert_{t= t_n-t_l} \left.\frac{d \Id_{[t_i,t_f]}}{dt} \right\vert_{t= t_n-t_l} \right. \\
& +  k ~\eta_\epsilon (t_n-t_l)\Big{ )} + O(\epsilon^{q+1})
\end{split}
\end{equation}

Then, the partition function is given by:
\begin{equation}
\begin{split}
\hat W(t_1,t_N) & = \int_{\setR^{N-1}} \prod_{n=1}^{N-1} dq_n e^{\ii S(\{q_n\})}  \\
& = \int_{\setR^{N-1}} \prod_{n=1}^{N-1} dq_n e^{- \ii \frac{\Delta t^2}{2} \vec \sum_{n,l} q_n q_l A(t_n,t_l) + o (\epsilon^{q+1})}  
\end{split}
\label{eq:partition_function_ho}
\end{equation}
with
\[
A(t_n,t_l) = m\left.\frac{d^2 \eta_\epsilon}{dt^2} \right\vert_{t= t_n-t_l} \Id_{[t_i, t_f]}(t_n) + m \left.\frac{d \eta_\epsilon}{dt} \right\vert_{t= t_n-t_l} \left.\frac{d \Id_{[t_i,t_f]}}{dt} \right\vert_{t= t_n-t_l}  +  k ~\eta_\epsilon (t_n-t_l).
\]
This can be considered as a matrix with elements $A_{n,l}$, and thus, Eq.~\eqref{eq:partition_function_ho} is a Gaussian integral. Using Wick's theorem, we recover the standard result~\cite{zinn2004path}:
\begin{equation}
\hat W =\left((-2 \ii \pi) ^{N/2}\Delta t^N \det(A)^{1/2}\right)^{-1}.
\end{equation}
\end{appendix}



\bibliographystyle{SciPost_bibstyle} 

\nolinenumbers

\end{document}